\pdfoutput=1
\documentclass[journal,onecolumn,12pt]{IEEEtran}
\usepackage{amsmath,amssymb,array,graphicx,url,tikz}
\usetikzlibrary{decorations.pathreplacing}
\newtheorem{theorem}{Theorem}[section]
\newtheorem{proposition}[theorem]{Proposition}
\newtheorem{corollary}[theorem]{Corollary}
\newtheorem{lemma}[theorem]{Lemma}
\newcommand{\eqdef}{\mathrel{:=}}
\newcommand{\Rrate}{\mathcal R}
\begin{document}
\title{Lempel--Ziv Parsing with Match-Length-Dependent Windows: Universality and Limits}
\author{Yingquan (Cody) Wu\thanks{Earlier versions of the window scheme appeared in the author's 2017 patent and 2021 Data Compression Conference (DCC) report~\cite{EarlierWindowPatent,WZIPWuDCC2021}. The codecs WLZ4 and WZIP built on it are described in~\cite{WZIPWuCodecs2026}, and their source code is public~\cite{WZIPCode2026}.}\thanks{Y. Wu is with Tenafe Inc., 591 W. Hamilton Ave, Campbell, CA 95008, USA
(e-mail: icodywu@gmail.com).}}
\maketitle
\begin{abstract}
LZ77 coders spend the same number of bits on the distance of every
match, although short matches rarely pay off at long distances. With
match-length-dependent windows, used in several practical codecs, a match
of length $\ell$ may refer back at most $W(\ell)$ symbols: short matches
get short distance fields, while long matches keep the full window $W$.
This breaks the suffix closure that classical LZ77 analyses rely on: the
tail of an admissible match need not be admissible. We bound the cost of
this loss, and of matches missed by a practical hash search, by the
phrase count of a single-window reference parse.
Consequently, when the full window starts at length $O(\log W)$, greedy
parsing keeps the convergence rates and the optimal finite-input rates of
LZ77 under substring-complexity constraints, and when it starts at length
$o(\log W)$, greedy parsing is universal for stationary ergodic sources.
We then delimit what the windows gain. A calibrated schedule never
increases the minimum code length and saves a polynomially small rate on
some iid sources. The leading redundancy term is set by the length code
instead: on uniform iid data, with logarithmic-cost lengths and a fixed
full-window threshold, it is at least of order $(\log\log W)/\log W$
whatever windows the short matches use, and a short codeword for a
growing longest match reduces it to $O((\log\log\log W)/\log W)$. We
also quantify the gains of coding single-symbol runs by count and of
Huffman coding the token fields, and obtain a consistent entropy-rate
estimator by selecting the window.
\end{abstract}

\begin{IEEEkeywords}
Data compression, Lempel--Ziv coding, LZ77, sliding window, universal
source coding, redundancy, Huffman coding, hashing, entropy estimation.
\end{IEEEkeywords}

\section{Introduction}
\label{sec:intro}
LZ77~\cite{LZ77} replaces a repeated substring by its length and its
backward distance to an earlier occurrence. A sliding window of $W$
symbols bounds the distance, so every distance field costs about
$\log_2W$ bits. This creates a familiar trade-off: a larger window finds
more distant repetitions, but it makes short matches so expensive that
they are often not worth coding. Match-length-dependent windows refine
the trade-off. A match of length $\ell$ may refer back at most
$W(\ell)\le W$ symbols, so short matches use small windows and short
distance fields, while long matches keep the full window. The decoder
reads the length first and infers the distance range from it, so no
extra signaling is needed. We call this family WLZ, where W stands for
multi-window, as in the codecs WZIP and WLZ4 built on
it~\cite{WZIPWuCodecs2026}. Restrictions of this kind appear in practice:
Lizard reserves its 24-bit offsets for matches of 16 bytes or
more~\cite{WZIPLizard2016}, Pucrunch restricts two-byte matches to
shorter distances than longer ones~\cite{WLZPucrunch}, and LZMA and
Brotli condition their distance models on the match
length~\mbox{\cite{WZIPLZMA2015,WZIPBrotli2016}}.

Length-dependent windows, however, remove a property on which the
classical analyses of LZ77 rest. Under a single window, every suffix of
an admissible match is again an admissible match at the same distance.
Comparisons of a parse with a reference parse, the optimality of greedy
parsing, and the reduction of optimal parsing to maximal matches all use
this suffix closure~\cite{WLZCrochemore2014,WZIPBitOptimal2013}. With
length-dependent windows it fails: in \texttt{abcdabcd} with $W(4)=4$
and $W(2)=2$, the second \texttt{abcd} is a match at distance four, but
its suffix \texttt{cd} is not. Two questions follow. Does restricting
short matches cost anything asymptotically, and what can the
restrictions gain?

This paper answers both questions for a broad family of window
schedules, parsers, and token codes. Its contributions are the
following.
\begin{enumerate}
\item \emph{A transfer theorem} (Theorem~\ref{thm:block-transfer}). Apart
from a term of order $(\log W)/W$ per symbol, the cost of a WLZ parse is
bounded by the phrase count of a single-window reference parse, times one token plus the literals of at most $t-1$
symbols per phrase, where $t$ is the length from which the full window
applies; each match missed by an inexact search adds one token. This
accounts exactly for the loss of suffix closure and makes the
established phrase-count bounds for LZ77 applicable to WLZ.
\item \emph{Guarantees that survive} (Section~\ref{sec:convergence}). If
the full window starts at length $O(\log W)$, exact greedy WLZ keeps the
convergence rates of LZ77 under substring-complexity constraints and, for
windows comparable to the input length, its finite-input minimax rates,
which are optimal (Theorem~\ref{thm:finite-minimax}). If the full window
starts at length $o(\log W)$, exact greedy WLZ is universal for every
finite-alphabet stationary ergodic source
(Theorem~\ref{thm:greedy-universality}). The same bounds hold for capped
hash-chain search on average over the hash choice, and on periodic inputs
an inspection cap of at least the period loses nothing
(Corollaries~\ref{cor:capped-search-transfer}
and~\ref{cor:piecewise-periodic}).
\item \emph{What the windows gain, and what they cannot}
(Section~\ref{sec:strict-comparisons} and
Appendix~\ref{app:obstruction}). A calibrated schedule never increases
the minimum code length and saves a polynomially small rate,
$\Omega(W^{-\beta}\log W)$ with $\beta>2$, on some iid sources
(Theorem~\ref{thm:strict-polynomial-saving}). The leading redundancy term
is out of the windows' reach: with logarithmic-cost length codes and a
fixed full-window threshold, it is at least of order
$(\log\log W)/\log W$ on uniform iid data, whatever windows the shorter
matches use (Theorem~\ref{thm:obstruction}). It is set by the length code
instead: a growing longest match with a short codeword reduces it to
$O((\log\log\log W)/\log W)$, below the $\Omega((\log\log W)/\log W)$ of
the original fixed-field LZ77 format with any phrase cap
(Theorem~\ref{thm:fixed-field-separation} and
Proposition~\ref{prop:original-any-cap}). Coding a run of one symbol by
its count, which may exceed the match cap, gives a further separation on
inputs with few runs (Proposition~\ref{prop:run-count}).
\item \emph{Coding the token fields} (Section~\ref{sec:field-coding} and
Theorem~\ref{thm:periodic-huffman}). Because the length is decoded first,
the distance code may depend on the length class at no signaling cost.
We bound the Huffman payload of separately coded fields and the loss of
coding range symbols with raw extra bits, as DEFLATE and Zstandard
do~\mbox{\cite{DEFLATE,Zstd}}
(Proposition~\ref{prop:significant-bit-locality}), and show that on
periodic inputs Huffman coding the fields lowers the rate from
$\Theta((\log W)/W)$ to at most $3/W$, tables included.
\item \emph{Window selection and exact parsing}
(Section~\ref{sec:finite-data} and Appendix~\ref{app:exact-parsing}).
Selecting the window by two-part coding yields an entropy-rate estimator
that is consistent almost surely and in mean, with finite-data error
bounds for nonuniform iid sources. For a given input, greedy parsing no
longer minimizes the number of phrases and optimal parsing no longer
reduces to maximal matches, yet the minimum code length is computable in
$O(n(K+Y)\log n)$ time, where $K$ and $Y$ count the intervals of lengths
on which the window and the length cost are constant.
\end{enumerate}

For codec design, the results separate the roles of the two components.
Below a full-window threshold of order $o(\log W)$, short matches can be
restricted arbitrarily at no asymptotic cost, which shortens their
distance fields and lets each length class use a small
hash index with a short seed (Section~\ref{sec:hash-search}). The
length-specific indexes can be searched in parallel without changing the
output, since the encoder selects among their candidates by length and
distance rather than by completion order, as in earlier parallel
length-specific designs~\cite{WLZFranaszek2006,WLZPlauth2023}; a preset
dictionary, as supported by zlib, Zstandard, and
LZ4~\cite{WZIPZlibManual,WZIPZstdAPI2025,WLZLZ4Dictionary2019}, simply
extends the history under the same limits. The leading compression
gains, however, must come from how lengths and distances are coded. The
codecs of~\cite{WZIPWuCodecs2026} are consistent with this division. On
the Silesia corpus, WZIP exceeds the compression ratio of Zstandard at
level~22 by 3.5--3.8\% (1.2--1.4\% without its input filters) while
decoding at 88\% of its speed, and WLZ4 exceeds that of LZ4HC by up to
16\%. In WLZ4, whose distance fields have fixed widths, the windows raise
the compression ratio; in WZIP, ablations indicate that length-conditioned
distance coding improves compression, while the restricted windows mainly
reduce indexing resources.

\emph{Related work.} The original LZ77 analysis counts phrases of
constrained languages~\cite{LZ77}; sliding-window and recurrence results
establish optimality for stationary ergodic
sources~\cite{WZIPWynerZiv1994,WZIPOrnsteinWeiss1993,WZIPJacobBansal2013},
phrase-count bounds under substring complexity give convergence
rates~\cite{WZIPRaskhodnikova2013,WZIPKociumaka2023}, mixing-based
hitting-time results give quantitative recurrence
tails~\cite{WZIPAbadi2021}, and a modified LZ scheme achieves
$O(1/\log W)$ redundancy under finite-memory
assumptions~\cite{WZIPWynerWyner1995}. On parsing, Crochemore, Langiu,
and Mignosi prove greedy phrase-count optimality under dynamic suffix
closure and uniform costs~\cite{WLZCrochemore2014}; minimum-bit parsing
handles variable costs~\cite{WZIPBitOptimal2013}, dictionary-symbolwise
flexible parsing is optimal under closure and coding
assumptions~\cite{WLZFlexibleParsing2012}, the gap between minimum-bit
and greedy parsing can be superconstant~\cite{WZIPKosolobov2018}, and
rightmost-occurrence and sublinear-time factorization address
construction~\cite{WZIPEllert2023,WLZEllertSublinear2023,WZIPKempa2024}.
Among multiple-index designs, Snappy couples a compact short-copy form
with an 11-bit offset~\cite{WZIPSnappy2011}, a recent GPU-decoding preprint
restricts ordinary matches by distance-dependent minimum
lengths~\cite{WZIPShavidze2026}, Sadakane and Imai combine two-level
hashing with suffix sorting~\cite{Two-Level-Hash}, Zstandard's
double-fast matcher uses two seed lengths under one window
limit~\cite{WZIPZstdDoubleFast2019}, Crochemore, Langiu, and Mignosi
keep one sliding-window suffix tree per class of equal offset
cost~\cite{WLZCrochemoreREP2013}, and context-selected dictionaries are a
distinct mechanism~\cite{WZIPPartialMatching1999}. The author's 2017
patent discusses conditional distance
ranges~\cite[Fig.~3 and eq.~(2)]{EarlierWindowPatent}, and the 2021
conference report~\cite{WZIPWuDCC2021} introduced the window scheme with
an LZ4-based experiment.

\emph{Organization.} Section~\ref{sec:background} fixes the parsing and
cost model, and Section~\ref{sec:windows} describes the windows and their
hash indexes. Section~\ref{sec:wzip-theory} gives the basic counting
bound, the coding of the token fields, and search costs.
Section~\ref{sec:convergence} proves the transfer theorem and the
guarantees that survive, Section~\ref{sec:strict-comparisons} the gains
and limits, and Section~\ref{sec:finite-data} window selection and
entropy estimation. The appendices contain the redundancy lower bound,
the proof of the significant-bit bound, and exact minimum-cost parsing.

\section{Parsing Model and Notation}
\label{sec:background}
This section fixes the objects to which all results refer: matches and
windows, the reference greedy parser, and the way code lengths are
counted.

\emph{Matches and windows.} Let $x_1^n$ be a sequence over a finite
alphabet. A match of length $\ell$ and backward distance $d$ at position
$i$ satisfies
\[
 x_{i+r}=x_{i-d+r}\quad(0\le r<\ell),\qquad 1\le d<i.
\]
Copies may overlap the current position, as in LZ77, and are decoded
symbol by symbol. A window of $W$ symbols permits distances
$1\le d\le W$. We distinguish window sizes, in symbols, from
distance-field widths, in bits: a length-$\ell$ match has window
$W(\ell)$ and width $w(\ell)\eqdef\log_2W(\ell)$, and the maximum window
$W$ has width $B\eqdef\log_2W$.

\emph{Notation.} We write $X_1^n$ for a random input, $x_1^n$ for a
realization, and $P_n$ for the law of $X_1^n$. Total code lengths are
denoted by $C$, coding rates by $\mathcal R$, header and framing costs by
$\Gamma$, and the entropy rate by $h$; $H$ denotes entropy, and $\eqdef$
introduces a quantity by a formula.

\emph{The reference parser.} The reference parser searches the entire
decoded history within the allowed distances. At each position it
chooses a longest admissible match of at least the minimum supported
length, or emits one literal if there is none, and then advances by the
length of the chosen phrase; at the end of the input, matches are
truncated. We call the phrases so chosen the \emph{conceptual phrases}
of the parse. Each conceptual phrase is coded as a match token or, if
that is cheaper, as its individual literals, with the same phrase
boundary in both cases. We call this the \emph{literal-fallback rule}.
It bounds the cost of a length-$\ell$ phrase by $b\ell$, where $b$ is
the cost of a literal, without changing the greedy boundaries.

\emph{Search.} Encoders find candidate matches with a chained hash
table: a head table gives the most recent position with each hash value,
and a chain links earlier positions with the same hash
(Section~\ref{sec:windows}). A candidate must be verified symbol by
symbol to reject hash collisions. The reference search inserts every
eligible position, including positions inside earlier matches, and never
discards a candidate that the longest admissible match needs. Practical
encoders also cap the number of candidates, keep only the latest hash
entry, skip insertions, or parse lazily, and their output generally
differs from the reference parse. Each result below therefore either
assumes exact greedy (or minimum-cost) parsing or accounts explicitly
for the matches that the search misses or declines.
Figure~\ref{fig:parsers} shows how greedy, lazy, and minimum-cost parses
of one input differ; minimum-cost parsing returns in
Section~\ref{sec:wzip-theory} and Appendix~\ref{app:exact-parsing}.

\begin{figure}[t]
\centering
\begin{tikzpicture}[font=\small,
  cell/.style={draw, minimum width=6mm, minimum height=6mm, inner sep=0pt, font=\ttfamily,
    text height=1.6ex, text depth=0.3ex},
  lit/.style={draw, fill=black!12},
  mat/.style={draw},
  row/.style={anchor=east, inner sep=0pt},
  cost/.style={anchor=west, inner sep=0pt}]
\node[anchor=west] at (-2.6,0.85) {history \texttt{abc\_bcdefg\_efgh\_pqr\_qrst\_stuvw\_}, then the input};
\foreach \c [count=\i] in {a,b,c,d,e,f,g,h,p,q,r,s,t,u,v,w}
  \node[cell] at (0.6*\i,0) {\c};
\newcommand{\phr}[5]{\draw[#5] (0.6*#2-0.3,#1-0.28) rectangle (0.6*#2+0.6*#3-0.3,#1+0.28);
  \node[font=\ttfamily\footnotesize, text height=1.5ex, text depth=0.3ex] at (0.6*#2+0.3*#3-0.3,#1) {#4};}
\node[row] at (0.0,-0.9) {greedy};
\phr{-0.9}{1}{3}{abc}{mat} \phr{-0.9}{4}{4}{defg}{mat} \phr{-0.9}{8}{1}{h}{lit}
\phr{-0.9}{9}{3}{pqr}{mat} \phr{-0.9}{12}{5}{stuvw}{mat}
\node[cost] at (10.1,-0.9) {89 bits};
\node[row] at (0.0,-1.65) {lazy};
\phr{-1.65}{1}{1}{a}{lit} \phr{-1.65}{2}{6}{bcdefg}{mat} \phr{-1.65}{8}{1}{h}{lit}
\phr{-1.65}{9}{1}{p}{lit} \phr{-1.65}{10}{4}{qrst}{mat} \phr{-1.65}{14}{3}{uvw}{mat}
\node[cost] at (10.1,-1.65) {87 bits};
\node[row] at (0.0,-2.4) {minimum cost};
\phr{-2.4}{1}{1}{a}{lit} \phr{-2.4}{2}{6}{bcdefg}{mat} \phr{-2.4}{8}{1}{h}{lit}
\phr{-2.4}{9}{3}{pqr}{mat} \phr{-2.4}{12}{5}{stuvw}{mat}
\node[cost] at (10.1,-2.4) {78 bits};
\end{tikzpicture}
\caption{\label{fig:parsers}Three parses of one input with a single window,
literals of 9 bits (shaded) and matches of 20 bits. Greedy parsing takes
the longest match at each position. Lazy parsing first emits a literal
when the next position offers a longer match: it gains on
\texttt{abcdefgh} but loses on \texttt{pqrstuvw}, where deferring to
\texttt{qrst} costs a third phrase. The minimum-cost parse, a shortest
path over all phrases, makes the better choice in each half.}
\end{figure}

\emph{Token formats.} LZSS permits literals as an alternative to a
match~\cite{LZSS}. DEFLATE codes literals and lengths in one alphabet and
distances in another, each with prescribed extra bits~\cite{DEFLATE};
LZ4 and Zstandard use other token layouts and entropy
coders~\mbox{\cite{LZ4,Zstd}}. To attribute a gain to the window
schedule, we either hold the token format fixed or include the
difference in the comparison.

\emph{Code lengths.} All results concern complete codes: headers,
schedules, initialization, and the tables of entropy-coded fields are
part of the code length. Token costs are fixed unless a field is entropy
coded. Large-file rates take the limit in the input length before the
limit in the window. Upper bounds hold for the stated source classes and
formats; they are not claimed to be best possible, and no redundancy
rate holds uniformly over all stationary ergodic
sources~\cite{WZIPShields1993}.

\section{Match-Length-Dependent Sliding Windows}
\label{sec:windows}
This section describes a window schedule and the hash indexes that
serve it. The analysis in later sections depends only on the schedule
$\ell\mapsto W(\ell)$; the indexes determine search cost and memory.

The length is decoded before the distance, so a schedule known to the
decoder determines the distance-field width without a class selector.
Table~\ref{tab:window-example} shows a schedule with one hash index per
length class. A match costs $a(\ell)+w(\ell)$ bits, where $a(\ell)$
covers the match flag and the length code, and it is cheaper than its
literals exactly when $a(\ell)+w(\ell)<b\ell$. A small window thus makes a
short match cheap, but the match still pays only if this inequality
holds. Let ${\cal F}_k$ hash a seed of $k$ symbols. An index that serves
a set of match lengths must use a seed no longer than the shortest of
them, so that a genuine match agrees on all hashed symbols.

\begin{table}[htbp]
\centering
\caption{Illustrative length-dependent windows and hash seeds.}
\label{tab:window-example}
\begin{tabular}{|c|rrrrrrr|}\hline
Match length $\ell$ (bytes)&2&3&4&5&6&7&$\ge8$\\\hline
Distance width $w(\ell)$ (bits)&5&9&12&15&17&19&21\\
Hash seed (bytes)&2&3&4&5&6&7&8\\\hline
\end{tabular}
\end{table}

\begin{figure}[t]
\centering
\begin{tikzpicture}[font=\small, >=stealth,
  occ/.style={circle, fill, inner sep=1.3pt},
  far/.style={circle, draw, fill=white, inner sep=1.3pt},
  head/.style={draw, inner sep=2pt, font=\footnotesize},
  lab/.style={anchor=east, align=right, inner sep=0pt, font=\footnotesize},
  brc/.style={decorate, decoration={brace, mirror, amplitude=3pt}},
  chain/.style={->, bend right=50}]
\node[anchor=west] at (-3.4,1.15) {(a) a single window};
\node[lab] at (-0.25,0.2) {all lengths\\3-byte seeds, $W$};
\draw (0,0) rectangle (9.5,0.4);
\draw[thick] (9.5,-0.15) node[below right, inner sep=1pt] {$i$} -- (9.5,0.62);
\foreach \x [count=\k] in {8.9,7.86,6.57,5.36,3.8,2.5,1.04} \node[occ] (a\k) at (\x,0.4) {};
\node[head] (ha) at (11.3,0.2) {$\mathcal F_3(\texttt{cat})$};
\draw[->] (ha.north) to[bend right=40] (a1);
\foreach \k [evaluate=\k as \n using int(\k+1)] in {1,...,6} \draw[chain] (a\k) to (a\n);
\draw[brc] (0,-0.05) -- (9.5,-0.05) node[midway, below=4pt]
  {every earlier \texttt{cat} is visited; every distance costs $B$ bits};
\node[anchor=west] at (-3.4,-1.45) {(b) length-dependent windows};
\node[lab] at (-0.25,-2.4) {$\ell=3$\\3-byte seeds, $W(3)$};
\draw[densely dotted] (0,-2.6) rectangle (7.43,-2.2);
\draw (7.43,-2.6) rectangle (9.5,-2.2);
\draw[thick] (9.5,-2.75) -- (9.5,-2.05);
\node[occ] (b1) at (8.9,-2.2) {}; \node[occ] (b2) at (7.86,-2.2) {};
\foreach \x in {6.57,5.36,3.8,2.5,1.04} \node[far] at (\x,-2.2) {};
\node[head] (hb) at (11.3,-2.4) {$\mathcal F_3(\texttt{cat})$};
\draw[->] (hb.north) to[bend right=40] (b1);
\draw[chain] (b1) to (b2);
\node[lab] at (-0.25,-3.7) {$\ell=5$\\5-byte seeds, $W(5)$};
\draw[densely dotted] (0,-3.9) rectangle (4.23,-3.5);
\draw (4.23,-3.9) rectangle (9.5,-3.5);
\draw[thick] (9.5,-4.05) -- (9.5,-3.35);
\node[occ] (c1) at (7.86,-3.5) {}; \node[occ] (c2) at (5.36,-3.5) {};
\node[far] at (1.04,-3.5) {};
\node[head] (hc) at (11.3,-3.7) {$\mathcal F_5(\texttt{cat\ a})$};
\draw[->] (hc.north) to[bend right=25] (c1);
\draw[->, bend right=35] (c1) to (c2);
\node[lab] at (-0.25,-5.0) {$\ell\ge8$\\8-byte seeds, $W$};
\draw (0,-5.2) rectangle (9.5,-4.8);
\draw[thick] (9.5,-5.35) -- (9.5,-4.65);
\node[occ] (d1) at (5.36,-4.8) {}; \node[occ] (d2) at (1.04,-4.8) {};
\node[head] (hd) at (11.3,-5.0) {$\mathcal F_8(\texttt{cat\ and\ })$};
\draw[->] (hd.north) to[bend right=14] (d1);
\draw[->, bend right=22] (d1) to (d2);
\end{tikzpicture}
\caption{\label{fig:indexes}Hash-chain search with one window and with
length-dependent windows (not to scale). Bars show the history before
the current position~$i$; dots mark earlier occurrences of each index's
seed, and arrows follow its chain from the bucket head, nearest first.
(a)~One index on 3-byte seeds visits every earlier \texttt{cat} within
$W$. (b)~One index per length class, keyed by the class's minimum
length and retaining only its window. Short seeds have many candidates
but only within a short window; hollow dots lie beyond it and are never
visited. The 8-byte seed \texttt{cat and } reaches the whole history
with few candidates.}
\end{figure}

Short and long matches affect the search differently. Many positions may
share the seed \texttt{cat}, but few share \texttt{cat and }. An index
with a longer seed avoids visiting candidates that cannot produce a long
match, and an index for short matches needs to retain only a short window
(Figure~\ref{fig:indexes}). The benefit depends on the source;
Section~\ref{sec:hash-search} bounds the expected number of candidates
and separates genuine seed matches from hash collisions.

Each index is a chained hash table whose chain is a circular array of
exactly $V$ entries for a window of $V$ symbols: the link of position $j$
occupies slot $j\bmod V$ (the bitwise $j\wedge(V-1)$ when $V$ is a power
of two) and is reused $V$ positions later, when $j$ leaves the window. A
dedicated length-$\ell$ index thus uses $W(\ell)$ link entries, and a
shared index retains the largest window it serves; the head table is
sized separately. In WZIP and WLZ4~\cite{WZIPCode2026}, each link stores
the distance to the previous position with the same hash, or a sentinel
if that position lies outside the window, and a search stops once the
accumulated distance leaves the window. Every slot it reads was therefore
written within the window, so only the head table needs initialization.
Links may instead store absolute positions, provided expired candidates
are rejected before a recycled slot is read. When distance $V$ is
admissible, the search at position $i$ precedes the insertion of $i$,
which reuses the slot of position $i-V$.

Each short class $\ell$ visits its candidates nearest first, verifies all
$\ell$ symbols, and stops at the first verified match, returning length
$\ell$ at the nearest retained distance. Only the class with the largest
window extends candidates beyond its seed and searches for the longest
admissible match, possibly up to an inspection cap. With a match cap of
$W$ and ties broken toward the nearest distance (\emph{nearest-distance
ties}), it may also stop at its first verified match of length $W$, since
no later candidate is longer or nearer. The encoder then selects among
the classes' candidates; a hit in one short class does not stop the
others. When distance costs vary, the nearest match need not be the
cheapest.

Adjacent length classes may share an index whose seed does not exceed
the shortest length it serves. Its chain must then retain the largest
window among those lengths, and each served length keeps its own stopping
condition: a length-3 hit does not end a shared length-3/4 search before
length 4 is resolved. Table~\ref{tab:shared-example} gives a shared-index
variant of Table~\ref{tab:window-example}; its 19-bit chain must retain
the whole 19-bit distance range for an exact search.

\begin{table}[htbp]
\centering
\caption{A shared-index variant of Table~\ref{tab:window-example}.}
\label{tab:shared-example}
\begin{tabular}{|c|c|c|c|}\hline
Length (bytes)&Distance width (bits)&Seed (bytes)&Chain width (bits)\\\hline
2&5&2&--\\
3&9&3&12\\
4&12&3&12\\
5&15&5&19\\
6&17&5&19\\
7&19&5&19\\
$\ge8$&21&8&21\\\hline
\end{tabular}
\end{table}

A length-2 class may use a direct lookup table instead of a chain; the
search is exact only if the lookup identifies the seed exactly. Whether
to admit length-2 matches at all depends on their net savings and on
their search and memory costs.

\section{Basic Bounds: Counting, Coding, and Search}
\label{sec:wzip-theory}
This section collects the basic tools used later.
Subsection~\ref{sec:counting} adapts the distinct-phrase counting of the
original LZ77 analysis to length-dependent windows,
Subsection~\ref{sec:single-window-comparison} compares a WLZ parse with a
single-window parse match by match, Subsection~\ref{sec:field-coding}
bounds the length of complete codes whose fields are Huffman coded, and
Subsection~\ref{sec:hash-search} bounds the work and memory of hash-chain
search. Throughout the paper, $\log$ denotes the base-two logarithm and
$\ln$ the natural logarithm.

A length-$\ell$ match at position $i$ and backward distance $d$ is
\emph{admissible} if
\[
  1\le d\le \min\{i-1,W(\ell)\},\qquad
  x_{i+t}=x_{i-d+t}\quad(0\le t<\ell).
\]
Unless stated otherwise, token costs are fixed: a literal costs $b$ bits
and a match costs
\begin{equation}
  c_W(\ell,d)\eqdef a(\ell)+w(\ell),
  \label{eq:wzip-cost}
\end{equation}
where $a(\ell)$ covers the match flag and a prefix code for the length,
and the distance $d-1$ is written in $w(\ell)\eqdef\log W(\ell)$ bits, an
integer. The decoded length determines the window, so no class
identifier is sent. A description of the schedule and any framing are
counted separately in $\Gamma_n$. When distances are Huffman coded, their
actual symbol and extra-bit costs replace $w(\ell)$.

\subsection{Counting short phrases}
\label{sec:counting}
The original LZ77 analysis counts distinct phrases~\cite[Sec.~III-C]{LZ77}.
With length-dependent windows, a phrase shorter than $m$ can start only
where no admissible length-$m$ match exists, so starts that share a
length-$m$ continuation must lie far apart. Let the parser be exact greedy
with literal fallback (Section~\ref{sec:background}), let $D_m(x)$ be the
number of distinct length-$m$ substrings of $x_1^n$, and let $S_m$ count
the conceptual phrases shorter than $m$.

\begin{lemma}[Spacing of short phrases]
\label{lem:wzip-spacing}
For $n\ge m$, with $m$ at least the minimum match length and no larger
than the maximum supported match length,
\begin{equation}
 S_m\le D_m(x)
 \left(1+\left\lfloor\frac{n-m}{W(m)+1}\right\rfloor\right)+m-1.
 \label{eq:wzip-spacing}
\end{equation}
\end{lemma}

\noindent\emph{Proof.}
Consider two short-phrase starts $i<j$ with complete length-$m$
continuations. If these continuations agree and $j-i\le W(m)$, an
admissible length-$m$ match exists at $j$, contradicting the selection of
a shorter phrase. Starts with the same continuation are therefore at least
$W(m)+1$ apart, so there are at most $1+\lfloor(n-m)/(W(m)+1)\rfloor$ of
them. There are $D_m(x)$ possible continuations, and at most $m-1$ phrase
starts have no complete continuation. Summing proves the claim.
\hfill$\square$

Let $r_m$ be the largest cost per symbol over supported lengths of at
least $m$:
\[
 r_m\eqdef\sup_{\substack{\ell\ge m\\ \ell\ \mathrm{supported}}}
 \frac{\min\{b\ell,a(\ell)+w(\ell)\}}{\ell}.
\]

\begin{corollary}[Finite-input rate bound]
\label{cor:wzip-rate}
Encode $x_1^n$ by exact longest-match parsing, representing each chosen
match by the cheaper of its token and its literals. For every supported
match length $m\le n$, the complete code length satisfies
\begin{align}
 \frac{C_W(x)-\Gamma_n}{n}
 &\le r_m+\frac{b(m-1)}{n}
 \left[D_m(x)
 \left(1+\left\lfloor\frac{n-m}{W(m)+1}\right\rfloor\right)+m-1\right].
 \label{eq:wzip-rate}
\end{align}
\end{corollary}
\noindent\emph{Proof.}
Phrases of length at least $m$ cost at most $nr_m$ bits in total. Shorter
phrases contain at most $(m-1)S_m$ symbols and cost at most $b(m-1)S_m$
bits. Apply Lemma~\ref{lem:wzip-spacing} to $S_m$. \hfill$\square$

If the search misses or declines matches, let $E_m$ count the starts of
short phrases that have a complete length-$m$ continuation and an
admissible length-$m$ match. Removing these starts restores the spacing
argument, so \eqref{eq:wzip-spacing} holds with $E_m$ added, and
\eqref{eq:wzip-rate} holds with $b(m-1)E_m/n$ added. This covers capped
searches, incomplete insertion, and lazy decisions.

\subsection{Comparison with a single window, and the choice of widths}
\label{sec:single-window-comparison}
The next lemma compares WLZ with a single window of $2^B$ symbols, match
by match. Let a comparator use that window, the same literal cost, and
match cost $a(\ell)+B$. Fix any comparator parse $P$, and partition its
matches into $\mathcal A$, admitted by the length-dependent windows, and
$\mathcal F$, excluded by them. Define
\[
 G(P)\eqdef\sum_{(\ell,d)\in\mathcal A}[B-w(\ell)],\qquad
 R(P)\eqdef\sum_{(\ell,d)\in\mathcal F}[b\ell-a(\ell)-B].
\]

\begin{lemma}[Reference-parse comparison]
\label{lem:reference-parse}
Let $C_W^*(x)$ be the minimum token cost of an admissible parse. Then
\begin{equation}
 C_W^*(x)\le C_0(P)-G(P)+R(P).
 \label{eq:wzip-oracle}
\end{equation}
\end{lemma}

\noindent\emph{Proof.}
Keep all admitted matches and original literals, and replace every
excluded match by its literals. Every later phrase sees the same decoded
history, so the result is an admissible parse whose cost is exactly the
right side, and the optimum costs no more. \hfill$\square$

If every comparator match is cheaper than its literals, then $R(P)\ge0$.
Let a heuristic WLZ encoder exceed the minimum token cost by
$E_W\eqdef C_W^{\rm heur}-C_W^*$, and let $\Delta\Gamma\eqdef\Gamma_W-\Gamma_0$
be the difference in header cost. The encoder is then strictly shorter
than the comparator parse whenever
\begin{equation}
 G(P)>R(P)+E_W+\Delta\Gamma.
 \label{eq:wzip-gain-condition}
\end{equation}
The saving $G(P)$ on admitted matches must therefore outweigh the loss
$R(P)$ on excluded ones. Minimum-bit parsing is a shortest-path problem
with an edge for every admissible phrase~\cite{WZIPBitOptimal2013};
Appendix~\ref{app:exact-parsing} shows which single-window shortcuts fail
under length-dependent windows and computes $C_W^*$ in $O(n(K+Y)\log n)$
time, where $K$ counts the intervals of lengths with a constant window and
$Y$ those with a constant length cost.

The same comparison gives a rule for choosing the widths. Fix a reference
parse with $N_\ell$ matches of length $\ell$, and let $A_\ell(w)$ count
those at distance at most $2^w$. Keeping the admitted matches and
replacing the others by literals costs
\[
 F_\ell(w)\eqdef A_\ell(w)[a(\ell)+w]
          +[N_\ell-A_\ell(w)]b\ell.
\]

\begin{proposition}[Marginal window expansion]
With $\Delta A\eqdef A_\ell(w+1)-A_\ell(w)$,
\begin{equation}
 F_\ell(w+1)-F_\ell(w)
 =A_\ell(w)-\Delta A[b\ell-a(\ell)-w-1].
 \label{eq:wzip-width-rule}
\end{equation}
\end{proposition}

\noindent\emph{Proof.}
Every previously admitted match costs one more bit. Each newly admitted
match saves $b\ell-a(\ell)-w-1$ bits relative to its literals.
\hfill$\square$

Widening the window therefore helps exactly when the newly admitted
matches save more than the extra bit costs the matches already admitted.
Nondecreasing widths for $J$ length classes can be chosen by the
recursion
\[
 T(j,w)\eqdef F_j(w)+\min_{v\le w}T(j-1,v),
\]
where the objective of a class sums those of its lengths; with prefix
minima it takes $O(J(B_{\max}+1))$ time once the histograms are built.
The rule optimizes a fixed parse; choosing widths and reparsing jointly
is a harder problem, and widths chosen per block must be included in the
code length.

\subsection{Entropy coding of the token fields}
\label{sec:field-coding}
Practical LZ77 codecs entropy code the fields of their tokens. This
subsection bounds the complete length of such codes, which
Sections~\ref{sec:convergence} and~\ref{sec:strict-comparisons} use.
Entropy and Huffman costs of LZSS fields have been studied before: Bell
analyzed lengths, distances, literals, and
flags~\cite[Sec.~6.3]{WLZBell1986}, and Fenwick entropy coded lengths
and distance precisions with raw distance
bits~\cite[Secs.~5 and~8]{WLZFenwick1995}. Figure~\ref{fig:fields} shows
the fields of a parse and the complete format used below.

\begin{figure}[t]
\centering
\begin{tikzpicture}[font=\small,
  cell/.style={draw, minimum width=7mm, minimum height=6mm, inner sep=0pt},
  box/.style={draw, minimum height=7mm, inner sep=2pt},
  brc/.style={decorate, decoration={brace, amplitude=3pt}},
  tag/.style={anchor=west, inner sep=0pt}]
\node[tag] at (-0.9,0) {(a)};
\foreach \c [count=\i] in {a,b,c,a,b,c,a,b,X,a,b,c}
  \node[cell, font=\ttfamily] (x\i) at (0.7*\i,0) {\c};
\draw[brc, decoration={mirror}] (x1.south west) -- (x3.south east)
  node[midway, below=4pt] {3 literals};
\draw[brc, decoration={mirror}] (x4.south west) -- (x8.south east)
  node[midway, below=4pt] {match $(5,3)$};
\draw[brc, decoration={mirror}] (x9.south west) -- (x9.south east)
  node[midway, below=15pt] {literal};
\draw[brc, decoration={mirror}] (x10.south west) -- (x12.south east)
  node[midway, below=4pt] {match $(3,6)$};
\node[tag] at (-0.9,-2.0) {(b)};
\foreach \y/\name/\what in {-2.0/V/literal values, -2.7/R/literal counts,
                            -3.4/L/match lengths, -4.1/D/offsets}
  \node[tag] at (0.35,\y) {$\name$\quad \what};
\foreach \c [count=\i] in {a,b,c,X} \node[cell, font=\ttfamily] at (3.5+0.7*\i,-2.0) {\c};
\foreach \c [count=\i] in {3,1,0} \node[cell] at (3.5+0.7*\i,-2.7) {\c};
\foreach \c [count=\i] in {5,3} \node[cell] (l\i) at (3.5+0.7*\i,-3.4) {\c};
\foreach \c [count=\i] in {3,6} \node[cell] (d\i) at (3.5+0.7*\i,-4.1) {\c};
\node[tag] at (7.3,-2.0) {$N=4$ symbols};
\node[tag] at (7.3,-2.7) {$M+1=3$; the last is the terminal count};
\node[tag] at (7.3,-3.4) {$M=2$};
\node[tag] at (7.3,-4.1) {$M=2$; the table may depend on the class $J$ of $\ell$};
\draw[->, dashed] (l2.east) .. controls (6.1,-3.75) .. node[right=1pt] {$J$} (d2.east);
\node[tag] at (-0.9,-5.7) {(c)};
\node[box, minimum width=40mm] (hd) at (2.05,-5.7) {$n,N,M,W$, mode bit};
\foreach \s [count=\i] in {V,R,L,D}
  \node[box, minimum width=13mm] (t\i) at (4.1+1.3*\i-0.65,-5.7) {$\mathcal T_{\s}$};
\foreach \s [count=\i] in {V,R,L,D}
  \node[box, minimum width=13mm] (p\i) at (9.3+1.3*\i-0.65,-5.7) {$\s$};
\draw[brc] (hd.north west) -- (hd.north east) node[midway, above=4pt] {self-delimiting header};
\draw[brc] (t1.north west) -- (t4.north east) node[midway, above=4pt] {labeled Huffman trees};
\draw[brc] (p1.north west) -- (p4.north east) node[midway, above=4pt] {payloads};
\end{tikzpicture}
\caption{\label{fig:fields}Fields of a parse and the complete format of
Section~\ref{sec:field-coding}. (a)~The greedy parse of
\texttt{abcabcabXabc} with nearest-distance ties: three literals, the
overlapping match $(5,3)$, a literal, and the match $(3,6)$. (b)~Its four
field streams. The decoder reads a length before its offset, so an offset
table may be selected by the length class~$J$. (c)~The complete format:
a self-delimiting header, four labeled Huffman trees, and the four
payloads, whose known sizes $N,M+1,M,M$ delimit them.}
\end{figure}

Because the length is decoded first, the distance code may depend on the
length class $J$ without a class selector. For a fixed parse with $M$
matches and distance symbol $Z$, the ideal empirical saving from
conditioning is
\begin{equation}
 M\widehat H(Z)-M\widehat H(Z\mid J)=M\widehat I(Z;J).
 \label{eq:wzip-mutual-information}
\end{equation}
Binary Huffman coding changes this saving by $\delta_g-\delta_c$, where
each rounding term lies in $[0,M)$ when a singleton alphabet may use zero
bits; extra distance bits, table headers, and format restrictions also
count. The identity explains the length-conditioned distance models of
LZMA and Brotli~\cite{WZIPLZMA2015,WZIPBrotli2016}, but it does not by
itself imply a net saving, nor a benefit from excluding matches.
Measurements agree: with the bounded dynamic-programming parser of WZIP,
replacing its per-length-group offset tables by one table lowers the
compression ratio on the Silesia corpus by 0.58\%, whereas admitting every
length at every distance, with offsets still coded per group, leaves it
unchanged~\cite{WZIPWuCodecs2026}. Each ablation reparses the input, so it
indicates rather than isolates the effect of conditioning.

We code a parse as four streams. Consecutive literals are grouped by a
count before each match, as in Zstandard~\cite[Sec.~3.1.1.3.2]{Zstd},
with an explicit terminal count, possibly zero, after the last match;
Zstandard instead appends the remaining literals after its final
sequence. In the next proposition, the fitted code tables, including
their integer symbol labels, are given to the decoder for free, and
lengths and offsets are coded as whole values. Let $C_{\rm H}^{\rm pay}$
denote the resulting Huffman payload, without table descriptions and
stream framing.

\begin{proposition}[Separate-field Huffman payload]
\label{prop:field-entropy}
Fix a valid parse with $N$ literals and $M$ matches, and put $P\eqdef N+M$.
Let $V,R,L,D$ be its literal-value, literal-count, match-length, and
offset streams, of respective sizes $N,M+1,M,M$. Define
\begin{equation}
 \mathcal E_4\eqdef N\widehat H_0(V)+(M+1)\widehat H_0(R)
              +M\widehat H_0(L)+M\widehat H_0(D),
 \label{eq:four-field-entropy}
\end{equation}
where empty streams contribute zero. Unconstrained empirical binary
Huffman coding satisfies
\begin{align}
 C_{\rm H}^{\rm pay}&\le\mathcal E_4+N+3M+1,
 \label{eq:four-field-huffman}\\
 C_{\rm H}^{\rm pay}&\le b_0N+P+1+2Mb_M
                   =O(P\log(P+1)),
 \label{eq:four-field-support}
\end{align}
where $b_0$ is a fixed valid literal width and
$b_M\eqdef\max\{1,\lceil\log\max\{1,M\}\rceil\}$.
The order assumes a fixed literal alphabet.
\end{proposition}
\noindent\emph{Proof.}
An empirical Huffman code exceeds the empirical entropy by at most one
bit per stream symbol, which gives~\eqref{eq:four-field-huffman}. For the
second bound, fixed-width literal coding is a valid competitor, and unary
count words $0^r1$ cost $\sum_j(r_j+1)=N+M+1=P+1$ bits. Each match-field
stream has at most $M$ distinct values, so a $b_M$-bit index is a valid
competitor for each value. Huffman optimality gives
\eqref{eq:four-field-support}. \hfill$\square$

The bounds need no independence between the streams. The support bound
relies on the free integer labels: a code that sends a range identifier
followed by raw extra bits must also pay for those bits.

A complete format for $n\ge1$ over a fixed alphabet $\Sigma$ sends
$n,N,M,W$ self-delimitingly, then the four labeled Huffman trees, then
the four payloads; a fixed mode bit tells whether lengths are capped at
$W$. Each tree has an empty-tree flag, and a nonempty tree is described
by preorder internal/leaf flags and a fixed-width label at each leaf. Let
$K_s$ be the number of distinct values in field $s$, with label universes
\[
 U_V\eqdef|\Sigma|,\quad U_R\eqdef N+1,\quad U_D\eqdef W,\quad
 U_L\eqdef\min\{n,W\}\ \text{or}\ n,
\]
according to the mode, where positive lengths and distances are labeled
from zero. A nonempty tree has $2K_s-1$ nodes, so the complete length
$C_{\rm H}$ satisfies
\begin{equation}
 C_{\rm H}-C_{\rm H}^{\rm pay}
 \le O(\log(n+1)+\log(W+1))
 +\sum_{s\in\{V,R,L,D\}}K_s(2+\lceil\log U_s\rceil).
 \label{eq:huffman-table-cost}
\end{equation}
The constant covers the empty-tree flags. The known stream sizes
$N,M+1,M,M$ delimit the payloads, including zero-bit codes for singleton
alphabets, so the format is a prefix code over all input lengths.
Reconstruction needs neither the frequencies, nor the search indexes, nor
the window schedule; any other decoder parameter must be fixed in advance
or included in the length.

These formats code a fixed parse; choosing the parse with the fitted
tables in mind is not covered by the fixed-cost parsing results. A
one-bit choice between the fitted format and another complete prefix
code $C_0$ costs $1+\min\{C_0,C_{\rm H}\}$ and so keeps every upper bound
proved for $C_0$ within one bit.

Zstandard codes its literal-length, match-length, and offset fields with
finite state entropy (FSE) coding, which carries a state between
symbols, interleaved with raw extra bits, and its literals with Huffman
coding~\cite[Secs. 3.1.1.3 and 4.1]{Zstd}. The results below are proved
for Huffman coding; carrying them over to FSE requires bounding its
normalization, final-state, and extra-bit costs.

DEFLATE and Zstandard shrink their coding alphabets by sending a range
identifier followed by prescribed raw bits
\cite[Sec.~3.2.5]{DEFLATE},~\cite[Sec.~3.1.1.3.2.1]{Zstd}. The next
proposition, a specialization of the analysis of grouped monotone
alphabets~\cite{WLZRyabko2003,WLZSaid2004}, bounds how much this loses
against coding whole values.

\begin{proposition}[Significant-bit coding with an irregular head]
\label{prop:significant-bit-locality}
Fix an integer $\nu\ge1$ and put $T\eqdef2^\nu$.
Let $X\in\{0,\ldots,c\}$, where $c\ge T$.
Give every $x<T$ and the cap $x=c$ an exact symbol. For $T\le x<c$,
encode the label $(r,u)$ followed by $r$ raw bits, where
\[
 r\eqdef\lfloor\log x\rfloor-\nu+1,
 \qquad u\eqdef\lfloor x/2^r\rfloor.
\]
Thus $\nu$ leading significant bits are retained. All labels, including
exact symbols, share one coding alphabet. Suppose
$p_k\ge\cdots\ge p_{c-1}$ for some $1\le k<T$, where
$p_x\eqdef\Pr(X=x)$; the head and cap probabilities are unrestricted.
For the label $J$, raw width $r_J$ (zero for exact symbols), and
$\tau\eqdef\Pr(k\le X<c)$, the grouping loss satisfies
\begin{equation}
 0\le\Delta\eqdef H(J)+\mathbb E r_J-H(X)
 \le\frac{2\tau}{e\ln2\,(2^\nu-k+1)}.
 \label{eq:significant-bit-locality}
\end{equation}
With optimal unrestricted Huffman codes for this same distribution,
the mean lengths obey $0\le L_{\rm group}-L_{\rm whole}<1+\Delta$.
An unused cap may have probability zero.
\end{proposition}
\noindent The proof is in Appendix~\ref{app:significant-bit-locality}.
For fixed $k$ the loss is $O(2^{-\nu})$, and the same holds for growing
$k$ if $2^\nu\ge2k$. Applied to the empirical distributions of $R,L,D$
that satisfy these conditions, with the literal stream unchanged,
\begin{equation}
 C_{\rm group}^{\rm pay}-C_{\rm H}^{\rm pay}
 \le 3M+1+(M+1)\Delta_R+M\Delta_L+M\Delta_D.
 \label{eq:significant-bit-streams}
\end{equation}
Again no independence between the fields is needed. With a common
precision, grouping adds $O((M+1)2^{-\nu})$ to the $O(M+1)$ rounding cost
of Huffman coding, so every rate order that already allows that rounding
is preserved. Table costs are not included.

The construction has $O(2^\nu[1+\log(c/2^\nu)])$ range labels: at most
$2^{\nu-1}$ per binary scale, plus the exact head and cap. For fixed
precision the alphabet shrinks from $c+1$ values to $O(\log c)$ symbols,
with smaller frequency tables and Huffman trees, while raw suffixes need
only bit extraction~\cite{WLZRyabko2003}. DEFLATE uses these ranges for
$D-1$ with $\nu=2$ and for $L-3$ with $\nu=3$, with an exact symbol for
$L=258$~\cite[Sec.~3.2.5]{DEFLATE}; for tails decreasing from $k=1$, the
grouping losses are below $0.266$ and $0.133$ bits per value. Zstandard's
literal-count, shifted-length, and coded-offset partitions refine dyadic
ranges~\cite[Sec.~3.1.1.3.2.1.1]{Zstd}, so the same argument bounds each
of their grouping losses by $1/(e\ln2)<0.531$ bits per value when the
distribution of positive values is nonincreasing.

\subsection{Hash-chain search and memory}
\label{sec:hash-search}
This subsection bounds the work of a hash-chain search and the memory of
the indexes. Fix a query seed of length $k$ and $N\le W$ stored seeds, of
which $M_k$ equal the query. The hash function is drawn, independently of
the input, from a family in which two unequal seeds collide with
probability at most $1/K$, where $K$ is the number of buckets. The
subscript $\mathrm{hash}$ on an expectation denotes averaging over this
choice, and the subscript $X$ averaging over the input as well.

\begin{lemma}[Expected bucket size]
\label{thm:wzip-hash}
If $Q_k$ is the number of prior nodes in the query bucket, then
\begin{equation}
 \mathbb E_{\mathrm{hash}} Q_k\le M_k+\frac{N-M_k}{K}.
 \label{eq:wzip-hash}
\end{equation}
For an iid source with symbol probabilities $p_a$ and collision entropy
$H_2\eqdef-\log\sum_a p_a^2$, at a fixed input position,
\begin{equation}
 \mathbb E_{X,\mathrm{hash}}Q_k
 \le \min(W,k-1)+W2^{-kH_2}+\frac WK.
 \label{eq:wzip-hash-iid}
\end{equation}
\end{lemma}

\noindent\emph{Proof.}
Equal seeds collide with probability one and unequal seeds with
probability at most $1/K$; summing the expectations of the indicators
proves \eqref{eq:wzip-hash}. Two nonoverlapping iid seeds agree with
probability $(\sum_a p_a^2)^k=2^{-kH_2}$, and at most $k-1$ earlier seeds
overlap the query, each contributing at most one to the expected number
of exact matches. Substitution gives \eqref{eq:wzip-hash-iid}.
\hfill$\square$

A search for a fixed length can stop at the first verified match, which
gives a sharper bound. For a class whose seed and target length are both
$\ell$, with history $W(\ell)$, let $R_\ell$ count the nodes inspected
before the first verified match, the end of the bucket, or an inspection
cap $S_\ell$. At a fixed position, let $N$ count the retained eligible
positions, of which $M_\ell$ have the same length-$\ell$ seed. With
$K_\ell$ buckets and the same hash assumption,
\[
 \mathbb E_{\mathrm{hash}} R_\ell
 \le \min\!\left\{S_\ell,
        \mathbf1_{\{M_\ell>0\}}+\frac{N-M_\ell}{K_\ell}\right\}
 \le \min\!\left\{S_\ell,1+\frac{N}{K_\ell}\right\}.
\]
Indeed, at most one node with an equal seed is inspected, and every other
inspected node is a hash collision; summing the collision indicators over
the retained positions gives the bound, with no source model and no
independence among collisions. With $K_\ell$ proportional to $W(\ell)$,
each short class thus inspects $O(1)$ candidates in expectation, even on
highly repetitive inputs: if all retained seeds equal the query, the
search visits one node instead of the whole bucket. Verification costs up
to $O(\ell)$ per node, and hashing and index updates cost extra. The bound
does not apply to the class with the largest window, where a seed match
shorter than the match cap must still be extended.

Shorter seeds also allow smaller windows at a comparable chance of
finding an exact seed match. For a uniform iid source on $q\ge2$ symbols,
fix a query position with $W_k$ preceding symbols, all retained. The query
word is independent of the preceding block, and any occurrence wholly
inside that block is admissible, so Lemma~\ref{lem:uniform-word-hit}
gives
\[
 \Pr\{\text{no exact $k$-symbol seed match within distance $W_k$}\}
 \le 2\exp\!\left\{-\frac{(\ln2)W_k}{8q^k}\right\}.
\]
A window of
\[
 W_k\ge \frac{8}{\ln2}\,q^k\ln\frac{2}{\varepsilon},
 \qquad 0<\varepsilon<1,
\]
symbols therefore misses with probability at most $\varepsilon$, and
shortening the seed by one symbol divides this window by $q$. Conversely,
for a random query, a union bound over distances, overlaps included,
bounds the probability of a seed match by $\min\{1,W_kq^{-k}\}$, so
coverage $1-\varepsilon$ requires $W_k\ge(1-\varepsilon)q^k$. Thus $q^k$ is
the right window order at fixed confidence; this converse averages over
the query word and does not match the dependence on $\varepsilon$. When a
class uses a seed shorter than its target length, the guarantee concerns
finding the seed, and the extension must still be verified.

For a nonuniform iid source, keeping $W_k2^{-kH_2}$ constant balances the
nonoverlapping term in \eqref{eq:wzip-hash-iid}, but does not give the same
guarantee for every query word. Genuine seed matches, unlike hash
collisions, supply useful candidates: they help a search that stops at the
first match, but lengthen an exhaustive scan of the bucket. A longer seed
thus offsets the larger number of genuine candidates in a larger window,
while the bucket count controls collisions. These bounds count
candidates, not time: on a constant string all seeds are equal, and
extension and index maintenance also cost time. Summing over all
positions gives a whole-file bound that does not rely on the
fixed-position calculation at the starts the parser selects.

With index groups $g$, chain windows $V_g$, bucket counts $K_g$, and entry
widths $s_g,t_g$ in bytes, the encoder needs
\[
 W_{\max}+\sum_g(s_gV_g+t_gK_g)
\]
bytes, plus lookahead, token buffers, entropy tables, and allocation
overhead; the circular layout of Section~\ref{sec:windows} uses exactly
$V_g$ link entries per index. The decoder needs one history buffer of size
$W_{\max}$. Chain windows that grow geometrically with ratio at least
$r>1$ have at most $rW_{\max}/(r-1)$ link entries in total, so with
$K_g=O(V_g)$ the head and chain arrays have $O(W_{\max})$ entries, whose
byte cost depends on $s_g$ and $t_g$. Comparisons with earlier accelerated
LZ77 matchers~\cite{Two-Level-Hash} should report these quantities.

\section{What Is Preserved: Rates, Minimax Orders, and Universality}
\label{sec:convergence}
This section shows that restricting short matches costs nothing
asymptotically, provided the full window starts early enough.
Subsection~\ref{sec:stronger-convergence} states the transfer theorem, the
main tool. Subsection~\ref{sec:rates} derives convergence rates under
substring-complexity constraints, Subsection~\ref{sec:capped-search}
extends the bounds to capped hash search, and
Subsection~\ref{sec:periodic} treats periodic inputs with Huffman-coded
fields. Subsection~\ref{sec:minimax} gives finite-input bounds and their
optimality, and Subsection~\ref{sec:greedy-universality} proves
universality for stationary ergodic sources.

An asymptotic statement concerns a family of encoders indexed by the
maximum window $W=2^B$, with $B\to\infty$. Every supported length
$\ell\ge t_B$ has the full window $W$, shorter lengths have windows of at
most $W$, and the \emph{full-window threshold} $t_B$ is at least the
minimum match length. The exact greedy family supports every length from
the minimum through at least $W$, and through all lengths where stated.
The literal cost $b$ and the length-code constants $a_0,a_1$ do not
depend on $B$ or $n$, and
\[
 a(\ell)\le a_0+a_1\log(\ell+1),\qquad a_0,a_1\ge0,
 \quad
 \Rrate_B(x)\eqdef\limsup_{n\to\infty}\frac{C_B(x_1^n)}n.
\]
Headers satisfy $\Gamma_n=o(n)$ for each fixed $B$, and literal fallback
keeps the greedy boundaries (Section~\ref{sec:background}). Rates take
the limit $n\to\infty$ first and then $B\to\infty$;
Corollary~\ref{cor:finite-polynomial} gives bounds that hold jointly for
finite $n$, and Section~\ref{sec:entropy-estimation} selects among a
growing set of windows. Two scales of convergence appear: for a
positive-entropy source, the redundancy above the entropy, often in
powers of $1/B=1/\log W$; for a zero-entropy or repetitive source, the
rate itself, which can decay as a power of $1/W$.

\subsection{The transfer theorem}
\label{sec:stronger-convergence}
The classical comparison of LZ77 with a reference parse uses suffix
closure: inside a reference copy, the rest of the copy is again a legal
copy at the same distance, so greedy parsing reaches its end. With
length-dependent windows this holds only while the remaining length is
long enough for the distance to be admissible. In \texttt{abcdabcd}, for
example, the second \texttt{abcd} is a match of length four at distance
four when $W(4)=4$, but its suffix \texttt{cd} cannot use that distance
when $W(2)=2$. The theorem below bounds the cost of this loss, and of the
failures of an inexact search. Let
$K_W\eqdef\max\{b,B+a_0+a_1\log(W+1)\}$; a conceptual phrase shorter
than $W$ then costs at most $K_W$, and a longer one at most $K_W/W$ per
symbol.

\begin{theorem}[Block-parse transfer]
\label{thm:block-transfer}\label{prop:schedule-search-transfer}
Retain the window and code assumptions above, with minimum match length
$m_0$, every length from $m_0$ through $W$ supported, and
$m_0\le t_B\le W$. Partition $x_1^n$ into consecutive blocks of length
$W$, with a final shorter block. Let the reference parse $P$ consist of
the ordinary greedy LZ77 parses of these blocks, each parsed
independently with every earlier position in that block available and
with overlapping copies allowed; let $Z\eqdef Z_W(x_1^n)$ be its phrase
count and $(\ell_j,d_j)$ its copy lengths and distances. Define
\begin{align*}
 \tau(d)&\eqdef\min\{k\in\{m_0,\ldots,W\}:
             d\le W(r)\text{ for every integer }r\in[k,W]\},\\
 \mathcal D(P)&\eqdef\sum_{j:\,\mathrm{copy}}
                  \min\{\ell_j,\tau(d_j)-1\}.
\end{align*}
Consider any valid parser that codes its phrases by the literal-fallback
rule. Let $F$ count its conceptual starts $i$ inside reference copies
$[p_j,q_j]$ where $r\eqdef q_j-i+1\ge\tau(d_j)$ but the chosen phrase has
length less than $r$. Then
\begin{equation}
 C_B^{\mathrm{heur}}(x_1^n)-\Gamma_n
 \le \frac{K_W}{W}n+K_W Z+b\mathcal D(P)+K_W F.
 \label{eq:schedule-search-transfer}
\end{equation}
Exact greedy search has $F=0$. Since $\tau(d)\le t_B$, so that
$\mathcal D(P)\le(t_B-1)Z$, exact greedy WLZ satisfies
\begin{equation}
 C_B(x_1^n)-\Gamma_n\le \frac{K_W}{W}n
 +[K_W+b(t_B-1)]Z_W(x_1^n).
 \label{eq:block-transfer}
\end{equation}
\end{theorem}

The reference parse serves only for comparison: the parser keeps its
whole history, and no minimum-bit parsing is needed. The count $F$
states exactly what an inexact search must preserve: from a start inside
a reference copy whose remainder is admissible, the chosen phrase should
reach the end of the copy. The parse may still differ from the exact
greedy parse.

\noindent\emph{Proof.}
The full window for long lengths ensures that $\tau(d)$ exists, and the
condition on every remaining length handles nonmonotone schedules.
Inside a reference copy, its suffix of length $r\ge\tau(d)$ is a legal
copy at the same distance, also when copies overlap, so exact greedy
parsing reaches the end of the reference copy.

For an arbitrary parser, every start in a reference copy other than the
last either counts toward $F$ or has fewer than $\tau(d)$ reference
symbols remaining. A failed phrase has length less than $r\le W$ and
costs at most $K_W$; it is necessarily not the last. The other phrases
that are not last stay within the final $\min\{\ell,\tau(d)-1\}$
positions of the reference copy and cover disjoint intervals there, so by
the literal-fallback rule they cost at most $b$ times this number. The
last start costs at most $K_W$ if its phrase is shorter than $W$. A
reference literal contains at most one start. Finally, all phrases of
length at least $W$ cost at most $K_W/W$ per symbol. A phrase crossing a
reference boundary is counted only where it starts, and a reference
phrase containing no start costs nothing. Summing
proves~\eqref{eq:schedule-search-transfer}.
\hfill$\square$

\subsection{Convergence rates under substring complexity}
\label{sec:rates}
The transfer theorem reduces convergence rates to phrase counts of
ordinary LZ77, which substring-complexity bounds control. We first
record a weaker bound that follows from Lemma~\ref{lem:wzip-spacing}
alone. For a threshold $m$ with $W_B(m)=W$, define
\[
 K_m\eqdef\max\{b,B+a_0+a_1\log(m+1)\}.
\]
Each conceptual phrase shorter than $m$ costs at most $K_m$, and each
longer phrase at most $K_m/m$ per symbol, because $1/\ell$ and
$\log(\ell+1)/\ell$ decrease for $\ell\ge1$. Lemma~\ref{lem:wzip-spacing}
therefore gives
\begin{equation}
 \frac{C_B(x_1^n)-\Gamma_n}{n}\le \frac{K_m}{m}
 +\frac{K_m}{n}\left[
 D_m(x_1^n)\left(1+\left\lfloor\frac{n-m}{W+1}\right\rfloor\right)
 +m-1\right].
 \label{eq:token-spacing}
\end{equation}
If $D_m(x_1^n)\le P(m)$ uniformly in $n$, this implies
\begin{equation}
 \Rrate_B(x)\le K_m\left(\frac1m+\frac{P(m)}{W+1}\right).
 \label{eq:token-rate}
\end{equation}
Counting a short phrase at the cost of a token rather than of its
literals matters when $m$ is much larger than $B$. For $P(m)=A m^s$ and
$t_B=O(B)$, choosing $m=\lfloor W^{1/(s+1)}\rfloor$ gives
$\Rrate_B=O(BW^{-1/(s+1)})$; the transfer theorem improves the exponent
to $1/s$ (Corollary~\ref{cor:polynomial-convergence}). For
stretched-exponential growth, the same bound already gives the following
rate.

\begin{proposition}[Stretched-exponential substring complexity]
\label{prop:stretched}
Assume $D_m(x_1^n)\le A2^{c m^\alpha}$ uniformly in $n\ge m$, for
fixed $A,c>0$ and $0<\alpha<1$. Under the same exact greedy family with
$t_B=O(B)$,
\begin{equation}
 \Rrate_B(x)=O\!\left(B^{-(1/\alpha-1)}\right).
 \label{eq:stretched}
\end{equation}
\end{proposition}

\noindent\emph{Proof.}
Choose a fixed $q\ge1/\alpha$ and set
$m=\lfloor((B-q\log B)/c)^{1/\alpha}\rfloor$ for sufficiently large $B$.
Then $m\ge t_B$, $m\le W$, $K_m=O(B)$, and
$P(m)/W\le A B^{-q}$. Equation~\eqref{eq:token-rate} gives
$O(B^{1-1/\alpha})+O(B^{1-q})$, proving the result. \hfill$\square$

The exponent of $1/\log W$ exceeds one when $\alpha<1/2$; for instance,
$\alpha=1/3$ gives $O((\log W)^{-2})$ for these zero-entropy sources. By
contrast, the exponential envelope $D_m\le A m^s2^{hm}$ with $h>0$, a
bound on language growth in the spirit of the constrained sources
of~\cite[Sec.~III-A]{LZ77}, gives only $\Rrate_B\le h+O(\log B/B)$ by
this argument: take $m=\lfloor[B-(s+2)\log B]/h\rfloor$ and assume
$t_B=o(B)$.

For power-law envelopes we combine the transfer theorem with a known
phrase-count bound. For a finite word $y$, let
\[
 \delta(y)\eqdef\max_{1\le m\le |y|}\frac{D_m(y)}m.
\]
The LZ77 phrase count $z(y)$ satisfies
\begin{equation}
 z(y)=O\!\left(\delta(y)\log\left(1+\frac{|y|}{\delta(y)}\right)\right),
 \label{eq:known-delta}
\end{equation}
a consequence of the alphabet-dependent bound of Kociumaka, Navarro, and
Prezza~\cite[Lemma II.4]{WZIPKociumaka2023}, which builds on
Raskhodnikova et al.~\cite{WZIPRaskhodnikova2013}; the $1$ in the
logarithm covers the case where $\delta(y)$ is comparable to $|y|$.

\begin{corollary}[Polynomial substring complexity]
\label{cor:polynomial-convergence}
Suppose $D_m(x_1^n)\le A m^s$ for every $n\ge m$, with fixed
$A\ge1$ and $s\ge1$. If $t_B=O(B)$, the exact greedy family satisfies
\begin{equation}
 \Rrate_B(x)\le
 \begin{cases}
 O(BW^{-1/s}),&s>1,\\
 O(B^2/W),&s=1.
 \end{cases}
 \label{eq:polynomial-convergence}
\end{equation}
The constants may depend on $A,s$ and the code-cost bounds.
\end{corollary}

\noindent\emph{Proof.}
For $s=1$, every reference block $y$ has $\delta(y)\le A$, so
\eqref{eq:known-delta} gives $z(y)=O(\log(1+|y|))$.
For fixed $s>1$, use the distinct-substring argument in the proof
of Claim 4 of~\cite{WZIPRaskhodnikova2013}. In a full length-$W$
reference block, let $n_k$ count ordinary greedy phrases of length
$k$, excluding the final phrase. For $1\le u\le W/2$,
\[
 \sum_{k\le u}k n_k\le D_{2u}+2u\le A(2u)^s+2u.
\]
Thus a dyadic length class $u/2<k\le u$ contains
$O(Au^{s-1}+1)$ phrases. Summing the geometric series up to $u$
gives $O(Au^{s-1}+\log u)$ phrases of length at most $u$.
There are at most $W/u$ longer phrases, plus the final phrase.
Take a power of two $u$ within a factor two of $W^{1/s}$; for
sufficiently large $W$ it is below $W/2$. Since $s>1$ is fixed,
$\log u=O(u^{s-1})$, and $z(y)=O(W^{1-1/s})$.

There are $\lfloor n/W\rfloor$ full reference blocks, while the
final block contributes at most $W$ phrases. Divide
\eqref{eq:block-transfer} by $n$ and let $n\to\infty$ at fixed $W$.
Since $K_W+b(t_B-1)=O(B)$, the two reference counts give
$O(BW^{-1/s})$ for $s>1$ and $O(B\log W/W)=O(B^2/W)$ for $s=1$.
\hspace*{\fill}$\square$

The exponent $1/s$ improves on the elementary $1/(s+1)$, and for $s>1$ the
dyadic argument also removes a logarithmic factor. Single-window LZ77
satisfies the same bounds; the point is that the windows do not destroy
them. Theorem~\ref{thm:finite-minimax} shows that the finite-input orders
are optimal, and polynomial rates for particular zero-entropy sources
have precedent~\cite{WZIPJainBansal2014}. Table~\ref{tab:convergence}
collects the guarantees of this subsection.

\begin{table}[htbp]
\centering
\caption{Convergence guarantees for the stated exact greedy family.
Constants in the source envelopes are fixed independently of $W,n$.}
\label{tab:convergence}
\small
\begin{tabular}{|p{2.25in}|p{2.3in}|}\hline
Substring-complexity envelope&Output-rate guarantee\\\hline
$D_m\le A m^s$, fixed $s>1$&$O((\log W)W^{-1/s})$\\[3pt]
$D_m\le A m$&$O((\log W)^2/W)$\\[3pt]
$D_m\le A2^{c m^\alpha}$, $0<\alpha<1$
 &$O((\log W)^{-(1/\alpha-1)})$\\[3pt]
$D_m\le A m^s2^{hm}$, $h>0$
 &$\le h+O(\log\log W/\log W)$\\\hline
\end{tabular}
\end{table}

\subsection{Capped hash search}
\label{sec:capped-search}
Practical encoders inspect a bounded number of candidates per index. For
such a policy, the failure count $F$ of Theorem~\ref{thm:block-transfer}
has an explicit bound. The \emph{capped multi-index policy} is as
follows. Fix index groups $g$ with seed lengths and extension limits
$1\le k_g\le L_g\le W$, histories $V_g\le W$, bucket counts $K_g$, and
integer inspection caps $S_g\ge1$, such that every supported $r\le W$ has
a queried group $g(r)$ with
\[
 k_{g(r)}\le r\le L_{g(r)},\qquad V_{g(r)}\ge W(r).
\]
Each index retains every eligible position, including positions inside
earlier phrases. At each conceptual start, the encoder queries every
group with a complete seed and inspects its bucket nearest first,
stopping after $S_g$ nodes or at the end of the bucket. It verifies the
seed and compares candidates exactly up to length $L_g$ or the end of the
input, considering every supported admissible length in $[k_g,L_g]$ that
it finds. A group with $L_g<W$ may stop at its first verified admissible
match of length $L_g$. With nearest-distance ties and a match cap of $W$,
a group with $L_g=W$ may likewise stop at its first verified admissible
match of length $W$; otherwise it continues to the cap or the end of the
bucket. The encoder selects a longest admissible match across the groups
and applies the literal-fallback rule.

The dedicated short-length indexes of Section~\ref{sec:windows} have
$k_g=L_g=\ell$ and $V_g=W(\ell)$, so each stops at its first verified
match; the largest class has $k_g=t_B$, $L_g=W$, and $V_g=W$. Shared
indexes satisfy the same coverage condition, but a hit shorter than their
longest served length does not justify stopping. The full-$W$ stop keeps
the nearest longest match, since later nodes in that group are older and
cannot give a longer admissible match; selection across groups still uses
all returned candidates. This stop does not apply to the parser with
unrestricted lengths of Corollary~\ref{cor:linear-entropy-payload}.

To state the bound, fix the input and the reference parse $P$. Let
$\mathcal Q(P)$ contain all positions $i$ in reference copies $[p_j,q_j]$
whose remaining suffix length $r_i\eqdef q_j-i+1$ is at least
$\tau(d_j)$. Fix a covering group $g_i\eqdef g(r_i)$, and let $d_i^*$ be
the distance of the nearest admissible full-length witness:
\[
 d_i^*\eqdef\min\{d:1\le d\le\min\{W(r_i),i-1\},\quad
                x_{i-d}^{i-d+r_i-1}=x_i^{i+r_i-1}\}.
\]
The reference distance $d_j$ is such a witness. Among the $d_i^*-1$
closer positions, let $A_i$ count the seeds equal to the query's
$k_{g_i}$-symbol seed, and define
\[
 v_i\eqdef\begin{cases}
  0,&S_{g_i}\ge d_i^*,\\
  1,&S_{g_i}\le A_i,\\
  \displaystyle\min\!\left\{1,
   \frac{d_i^*-1-A_i}{K_{g_i}(S_{g_i}-A_i)}\right\},&\text{otherwise}.
 \end{cases}
\]

\begin{corollary}[Capped multi-index search]
\label{cor:capped-search-transfer}
Use the capped multi-index policy, with hashes drawn independently of the
input as in Lemma~\ref{thm:wzip-hash}. Then
$\mathbb E_{\mathrm{hash}}F\le\sum_{i\in\mathcal Q(P)}v_i$ and
\begin{equation}
 \mathbb E_{\mathrm{hash}} C_B^{\mathrm{heur}}(x_1^n)-\Gamma_n
 \le \frac{K_W}{W}n+K_W Z+b\mathcal D(P)
                       +K_W\sum_{i\in\mathcal Q(P)}v_i.
 \label{eq:capped-search-transfer}
\end{equation}
The policy visits at most $n\sum_gS_g$ candidate nodes, and the history,
lookahead, and indexes use $O(W+\sum_g(V_g+K_g))$ symbols and entries,
apart from entry widths.
\end{corollary}

\noindent\emph{Proof.}
The covering group retains the witness and tests its required length.
Let $Z_i$ count the closer unequal seeds whose hash collides with the
query's. The witness has bucket rank $1+A_i+Z_i$, and
$\mathbb E_{\mathrm{hash}}Z_i\le(d_i^*-1-A_i)/K_{g_i}$. If the group stops
earlier at a verified admissible match of length $L_{g_i}$, that match
already reaches at least $r_i$; otherwise a witness within the inspection
budget is visited, and exact verification reaches $r_i$. Thus, when
$S_{g_i}>A_i$, a failure to reach the end of the reference copy requires
$Z_i\ge S_{g_i}-A_i$, and Markov's inequality gives the third branch of
$v_i$; the other branches follow from the rank being at most $d_i^*$ and
from the trivial bound on probabilities. No independence among
collisions is needed, and since every counted failure lies in the fixed
set $\mathcal Q(P)$, none is needed at the starts the parser selects
either. Theorem~\ref{thm:block-transfer} completes the proof; the visit
and storage bounds follow from the caps and the retained arrays.
\hfill$\square$

If $\mathcal G_{\rm short}$ consists of the dedicated fixed-length groups,
their first-match rule sharpens the expected total number of visits to
\[
 n\sum_{g\in\mathcal G_{\rm short}}
       \min\{S_g,1+V_g/K_g\}
 +n\sum_{g\notin\mathcal G_{\rm short}}S_g.
\]
This follows by summing the fixed-position work bound of
Section~\ref{sec:hash-search} over all positions, of which the parse
queries only a subset, so starts that depend on the hash need no
independence assumption. The worst-case bound by the caps remains valid.

For comparison, with $I_g\eqdef\{1,\ldots,n-k_g+1\}$,
$N_{g,i}\eqdef\min\{V_g,i-1\}$, and $M_{g,i}$ the number of retained equal
seeds, the bucket-size bound gives the coarser alternative
\[
 \mathbb E_{\mathrm{hash}}F\le U\eqdef\sum_g\sum_{i\in I_g}
 \min\left\{1,\frac{M_{g,i}+(N_{g,i}-M_{g,i})/K_g}{S_g+1}\right\}.
\]
Indeed, a failure requires a covering bucket to overflow; apply
\eqref{eq:wzip-hash} at the integer threshold $S_g+1$. The rank bound of
Corollary~\ref{cor:capped-search-transfer} instead counts only the
candidates ahead of a witness. On a constant string it gives zero with a
cap of one, even when the bucket-size bound is linear in $n$.

Suppose that, outside $E$ exceptional qualifying positions,
$d_i^*\le D_{g_i}$ and $A_i\le S_{g_i}/2$. Then
\begin{equation}
 \frac{\mathbb E_{\mathrm{hash}}F}{n}
 \le \min\!\left\{1,\frac En+
          \max_g\frac{2D_g}{K_gS_g}\right\}.
 \label{eq:witness-rank-local}
\end{equation}
These are explicit conditions on recurrence and partial seeds; we do not
claim that polynomial substring complexity implies them. For one
covering index with $V=K=W$, $S=\lceil W^{1/s}\rceil$, and fixed $s>1$,
the condition $A_i\le S/2$ at every qualifying position makes the added
rate $O(BW^{-1/s})$. Inputs that also satisfy the hypotheses of
Corollary~\ref{cor:finite-polynomial}, with $W\le n$, therefore satisfy
$(\mathbb E_{\mathrm{hash}} C_B^{\mathrm{heur}}-\Gamma_n)/n=O(BW^{-1/s})$
with $S=o(V)$, $O(nW^{1/s})$ visits, and $O(W)$ entries. The expectation
is for each finite input, not for one fixed choice of hashes.

Periodic inputs need no condition on partial seeds. On a globally
$p$-periodic input with $p\le W$ and all covering caps at least $p$, the
nearest witness has distance at most $p$: use the reference distance
$d_j$ if $d_j<p$, and distance $p$ otherwise. Hence $F=0$ for every hash
choice. Counting directly within periodic regions gives a stronger bound
that also allows the pattern to change.

\begin{corollary}[Piecewise-periodic inputs]
\label{cor:piecewise-periodic}
Use the policy and code-cost assumptions of
Corollary~\ref{cor:capped-search-transfer}, with $W\ge2$.
Suppose $x_1^n$ is a concatenation of $J$ nonempty intervals $I_j$,
each with a period $p_j$ satisfying
$1\le p_j\le\min\{|I_j|,p_*\}$, where $p_*\le W$.
Put $\Pi\eqdef\sum_{j=1}^J p_j$ and require every covering index cap to
satisfy $S_g\ge p_*$. Then, for every hash choice,
\begin{equation}
 C_B^{\mathrm{heur}}(x_1^n)-\Gamma_n
 \le \frac{K_W}{W}n+b\Pi+J\{b(t_B-1)+2K_W\}.
 \label{eq:piecewise-periodic}
\end{equation}
In particular, bounded literal cost and $t_B=O(B)$ give payload cost
$O(nB/W+\Pi+JB)$. The segmentation and the periods serve only the
analysis; neither the encoder nor the decoder needs them.
\end{corollary}

\noindent\emph{Proof.}
Every conceptual phrase of length at least $W$ costs at most $K_W/W$ per
symbol, which gives the first term. For an interval $I_j=[a_j,e_j]$,
consider a selected start $i\ge a_j+p_j$ with at least $t_B$ interval
symbols remaining, and set $r\eqdef\min\{W,e_j-i+1\}$. Periodicity
supplies a length-$r$ copy at distance $p_j$, and since $r\ge t_B$, its
full window makes that distance admissible. A covering index retains the
witness, whose bucket rank is at most $p_j\le S_g$ whatever the
collisions. The search either visits this witness or stops earlier at a
verified admissible match of length $L_g\ge r$; either way it finds a
match of length at least $r$. Consequently, any selected phrase shorter
than $W$ at such a start reaches the end of the interval.

Assign each remaining short phrase to the interval containing its start.
Those contained in the first $p_j$ positions together cost at most
$bp_j$ by literal fallback, and at most one short phrase starts in that
prefix and crosses its end, costing at most $K_W$. After the prefix, a
short phrase that does not reach the end of the interval must start
within its final $t_B-1$ positions. Such contained phrases cost at most
$b(t_B-1)$ in total, and at most one further short phrase reaches or
crosses the end of the interval, costing at most $K_W$. Overlapping
prefix and terminal regions only increase this upper bound, and a phrase
entering from a preceding interval only removes starts. Summing
proves~\eqref{eq:piecewise-periodic}. \hfill$\square$

For one globally $p$-periodic input, this yields
\[
 \frac{C_B^{\mathrm{heur}}-\Gamma_n}{n}
 =O\!\left(\frac{B}{W}+\frac{p+B}{n}\right).
\]
The period is paid for once, so the large-file rate of the payload at
fixed $W$ is $O(B/W)$. With $W$ comparable to $n$, the piecewise bound is
$O(\Pi+(J+1)\log n)$ bits, apart from headers. One full-history index
with seeds of length one and $S=p_*$ uses at most $np_*$ visits and
$O(W)$ entries, even with one bucket, so $p_*=o(W)$ permits a sublinear
inspection cap. For a fixed cap, the same output satisfies the bound for
every qualifying segmentation, so the bound may be minimized over them
without transmitting one. The argument uses ordinary copies and supports
lengths through $W$; the run format of Section~\ref{sec:run-escape}, with
its logarithmic cap, supports different lengths.

These bounds count candidate visits. Extending a match symbol by symbol
can cost $O(W)$ per candidate, and hashing and index maintenance also
cost time. A single seed of length $k=t_B$ covers only suffixes of length
at least $t_B$; the finer bound with $\mathcal D(P)$ requires the coverage
specified above.

\subsection{Periodic inputs with Huffman-coded fields}
\label{sec:periodic}
On a periodic input, the parse consists almost entirely of matches of
length $W$ at one distance, so Huffman coding its fields can make each
token cost about three bits. The next theorem makes this precise for
complete codes.

\begin{theorem}[Complete periodic Huffman coding under capped search]
\label{thm:periodic-huffman}
Use the policy of Corollary~\ref{cor:piecewise-periodic} on a globally
$p$-periodic input, with $p\le W$ and a match cap of $W$.
Put $S_{\max}\eqdef\max_gS_g$ and $D_*\eqdef\min\{W,pS_{\max}\}$.
For $W$ large enough that $K_W<bW$, fit the four codes of
Proposition~\ref{prop:field-entropy} to the emitted parse, keeping its
literal fallback, and transmit the tables and framing
of~\eqref{eq:huffman-table-cost}. Then
\begin{equation}
 \limsup_{n\to\infty}\frac{C_{\rm H}(x_1^n)}n
 \le\frac{3+\lceil\log D_*\rceil}{W}.
 \label{eq:periodic-entropy-payload}
\end{equation}
If ties among visited longest matches are resolved by smallest distance,
the sharper complete-code bound is
\begin{equation}
 \limsup_{n\to\infty}\frac{C_{\rm H}(x_1^n)}n\le\frac3W,
 \label{eq:periodic-nearest-complete}
\end{equation}
without an upper bound on the inspection caps. With this tie rule,
$B=\log W$, and $t_B=O(B)$, the finite-input bound for $n\ge W$ is
\begin{equation}
 C_{\rm H}(x_1^n)\le\frac{3n}{W}
                +O((p+B)B+\log(n+1)),
 \label{eq:periodic-nearest-finite}
\end{equation}
with fixed alphabet and code-cost constants.
\end{theorem}
\noindent\emph{Proof.}
Every visited candidate has distance at most $D_*$. To see this, fix a
queried group $g$. If its available history is shorter than $pS_g$, all
candidates have smaller distance. Otherwise the $S_g$ positions at
distances $p,2p,\ldots,S_gp$ have the same seed as the query, overlapping
seeds included. With complete retention and nearest-first traversal, the
first $S_g$ bucket entries lie within that range whatever the
collisions, and early stopping only shortens this prefix. Distances also
never exceed $W$, and groups without a complete query seed contribute no
candidates.

The proof of Corollary~\ref{cor:piecewise-periodic} bounds the number of
conceptual phrases shorter than $W$ by $Q\eqdef p+t_B+1$, and the number
of emitted literals by
\[
 N\le p+t_B-1+2K_W/b.
\]
The condition $K_W<bW$ prevents literal fallback on a full-length phrase.
Thus all but finitely many matches have length $W$, all but finitely
many literal counts are zero, and $M/n\to1/W$. For each of these two
fields, give the dominant value a one-bit word and code the finitely
many exceptional values under the other branch; together they cost at
most $2M+O_W(1)$. A fixed code for offsets costs at most
$M\max\{1,\lceil\log D_*\rceil\}$, and literals cost $O_W(1)$. All four
supports and their labels are bounded at fixed $W$, uniformly in $n$, so
the tables and framing of~\eqref{eq:huffman-table-cost} cost
$O_W(1)+O(\log(n+1))$ bits. Huffman optimality and division by $n$
prove~\eqref{eq:periodic-entropy-payload} for the complete code.

For nearest ties, let $p_0$ be the primitive period of the infinite
extension of the period word, so $p_0\mid p$. A full-$W$ equality, with
$W\ge p$, covers every residue modulo $p$; its shift therefore preserves
the cyclic period word and is a multiple of $p_0$. Whenever a full-$W$
match can occur, distance $p_0$ is available. A group covering length
$W$ has $L_g=W$, and the witness has bucket rank at most
$p_0\le p\le S_g$. No closer candidate is a full-$W$ match, since its
distance would be a smaller positive multiple of $p_0$, so the witness is
visited even when the search stops at the first full match. Nearest ties
select $p_0$ for every full-$W$ match, so the offset field also has one
dominant symbol and finitely many exceptions. Giving each dominant
symbol one bit proves~\eqref{eq:periodic-nearest-complete}.

For the finite bound, $Q,N=O(p+B)$ and $M\le n/W+Q$. Each of the three
integer fields has at most $Q$ occurrences outside its dominant symbol,
and prefixing an index for each exceptional value gives a total integer
payload of $3M+O(Q\log(Q+1))$. The literal payload is $O(N)$. The tree
labels cost $O((Q+1)B)$, since $p\le W$ and all literal counts are at most
$N=O(p+B)$. Adding the framing and using $\log(Q+1)=O(B)$
proves~\eqref{eq:periodic-nearest-finite}.
\hfill$\square$

The decoder receives no period and no special token for periodic
patterns. To see the gain, compare two codes for the same parse with
nearest-distance ties. For $W=2^B$, let $C_{\rm fix}$ code the parse with
its fixed-width distance fields and a match cost of $B+a(W)$ for a
length-$W$ copy, where $0\le a(W)=O(B)$, and with complete header cost
$o(n)$ at fixed $W$. The proof gives $M_W/n\to1/W$ for the number $M_W$ of
length-$W$ copies, with only $O_W(1)$ other tokens. Consequently,
\[
 \lim_{n\to\infty}\frac{C_{\rm fix}(x_1^n)}n
   =\frac{B+a(W)}W=\Theta(B/W),
 \qquad
 \limsup_{n\to\infty}\frac{C_{\rm H}(x_1^n)}n\le\frac3W.
\]
For a fixed periodic source, Huffman coding thus saves a logarithmic
factor as $W\to\infty$ after the large-file limit. The comparison
concerns two codes for one parse; the fixed-field format of the original
LZ77 paper is compared in Theorem~\ref{thm:fixed-field-separation}. The
gain requires the dominant values to have bounded extra-bit cost: a raw
extra-bit field that grows with $W$ keeps a logarithmic cost even when
its range code is predictable.

With the optional full-$W$ stop, each full-length covering group
inspects at most $p_0$ nodes at a start that admits a full-$W$ match; on
a constant input this is one visit. For bounded $p$, the finite
remainder is $O(B^2+\log(n+1))$, which must still be counted when $W$
grows with $n$. The rate amortizes one set of fitted tables over the
whole input; refreshing the tables in successive blocks pays for them
repeatedly. Without nearest ties, a one-bucket index with cap $W$ on a
constant input can cycle through offsets $1,\ldots,W$, keeping offset
entropy $\log W$ and a rate of order $(\log W)/W$, so the tie rule
controls coding entropy as well as match selection.

The constant $3$ reflects symbolwise Huffman coding, not a limitation of
the parse. For $n\ge W$ and nearest ties, the integer streams have
default values $0,W,p_0$ and at most $Q$ exceptions each. Enumerative
coding~\cite{WLZCover1973} describes a length-$m$ stream over $U$ labels
with $s$ exceptions in
\[
 O(\log U+\log(s+1)+1)
 +\left\lceil\log\binom ms\right\rceil
 +s\left\lceil\log\max\{1,U-1\}\right\rceil
\]
bits: send the default value and the number of exceptions, rank the
exceptional positions, and send their values; an empty stream or a
singleton alphabet costs zero. Send $n,N,M,W$ self-delimitingly and the
$N$ literals at fixed width; the integer universes are $N+1$,
$\min\{n,W\}$, and $W$. At fixed $W,p$, the universes and the numbers of
exceptions are bounded, so the complete length is
$O_{W,p}(\log(n+1))$, with all default values, including $p_0$,
included. This changes only the representation of the same parse, and the
decoder again receives no period.

\subsection{Finite-input bounds and their optimality}
\label{sec:minimax}
The rates of Subsection~\ref{sec:rates} take the input-length limit
first. The same proof gives bounds that hold uniformly for finite inputs,
once the last reference block is counted by its substring complexity
rather than by its length. These bounds are optimal: no prefix code does
better on the worst input of the class, and WLZ attains this optimum for
every window schedule in a broad family.

\begin{corollary}[Finite-input complexity bound]
\label{cor:finite-polynomial}
Use the exact greedy family of Theorem~\ref{thm:block-transfer}, with
$t_B=O(B)$, $W=2^B\le n$, and supported lengths through at least $W$.
Suppose $D_m(x_1^n)\le A m^s$ for all $1\le m\le n$, with fixed
$A\ge1$ and $s\ge1$. Uniformly for sufficiently large $W\le n$,
\begin{equation}
 \frac{C_B(x_1^n)-\Gamma_n}{n}\le
 \begin{cases}
 O(BW^{-1/s}),&s>1,\\
 O(B\log(1+W)/W),&s=1.
 \end{cases}
 \label{eq:finite-polynomial}
\end{equation}
Consequently, choosing the largest power of two at most $n$ gives
$C_B/n\le \Gamma_n/n+O((\log n)n^{-1/s})$ for $s>1$, and
$C_B/n\le \Gamma_n/n+O((\log n)^2/n)$ for $s=1$.
\end{corollary}

\noindent\emph{Proof.}
Write $n=aW+r$, $0\le r<W$. For $s>1$, the dyadic argument in
Corollary~\ref{cor:polynomial-convergence} bounds the ordinary phrase
count of every length-$m$ reference block by $c m^{1-1/s}$, since every
such block inherits the assumed substring bound. Hence
\[
 Z_W\le c\{aW^{1-1/s}+\mathbf1_{\{r>0\}}r^{1-1/s}\}.
\]
Substitution in~\eqref{eq:block-transfer}, whose coefficients are
$O(B)$, gives $O(BW^{-1/s}+BW^{1-1/s}/n)$. For $s=1$, use
$z(y)\le c\log(1+|y|)$ from~\eqref{eq:known-delta} to obtain
$O(B\log(1+W)/W+B\log(1+W)/n)$. In each case the second term is at most
the first because $W\le n$. \hfill$\square$

The encoder starts with an empty history and compresses from the first
symbol; no raw prefix of length $W$ is needed. Within this bound, the
largest allowed window gives the strongest guarantee, so there is no
intermediate optimal window. The constants are uniform in $n$ and $W$ but
depend on $A$, $s$, and the code. The header cost $\Gamma_n$ must be
controlled along the chosen $W(n)$; a bound $o(n)$ for each fixed $W$
does not suffice. The next theorem shows that these orders are optimal,
for complete code lengths.

\begin{theorem}[Finite-input minimax order]
\label{thm:finite-minimax}
Fix $A\ge2$ and $s\ge1$, independently of $n$, and let
\[
 \mathcal F_{n,s,A}\eqdef\{x\in\{0,1\}^n:
       D_m(x)\le A m^s\text{ for all }1\le m\le n\}.
\]
Let $C$ range over injective maps from $\{0,1\}^n$ to finite binary
strings whose images are prefix-free, with $n$ known to the decoder.
Then, as $n\to\infty$,
\begin{equation}
 L^*_{n,s,A}\eqdef\inf_C\max_{x\in\mathcal F_{n,s,A}}|C(x)|
 =\begin{cases}
   \Theta(n^{1-1/s}\log n),&s>1,\\
   \Theta((\log n)^2),&s=1.
  \end{cases}
 \label{eq:finite-minimax}
\end{equation}
The constants may depend on $A,s$. A single exact greedy WLZ format,
without knowledge of $A$ or $s$, attains both upper orders.
\end{theorem}

\noindent\emph{Proof.}
For the upper bound, take $n\ge4$, $B=\lfloor\log n\rfloor$, and
$W=2^B$. Use minimum match length two and the decoder-known schedule
\[
 W_B(2)\eqdef1,\qquad W_B(\ell)\eqdef W\quad(\ell\ge3).
\]
Support all feasible lengths and overlapping copies, and parse by exact
longest admissible matching with a fixed tie rule. Code a binary literal
$a$ as $0a$, and a match as a $1$ flag, the Elias gamma code of $\ell-1$,
and $d-1$ in exactly $\log W_B(\ell)$ bits. The gamma code reveals the
distance-field width; length-two matches have distance one and no
distance bits. The literal cost is $b=2$, and a match costs at most
\[
 \log W_B(\ell)+2+2\log(\ell+1).
\]
Apply the literal-fallback rule. Decoding stops after $n$ symbols, so the
complete codewords form a prefix code. The schedule and $B$ are
determined by $n$, so no schedule table or fitted model is sent, and the
header is constant for a shared $n$, or $O(\log n)$ bits when $n$ is sent
self-delimitingly. Small $n$ can be coded literally. This format has
$t_B=3=O(B)$ and logarithmic length costs, so
Corollary~\ref{cor:finite-polynomial}, with $n/2<W\le n$, gives both
upper bounds for the complete code.

For the lower bound when $s>1$, put
\[
 d=\lceil n^{1/s}\rceil,\qquad k=\lfloor n/(2d)\rfloor.
\]
Partition the first $2dk$ positions into $k$ blocks of length $2d$. In
each block independently place one $1$ in one of its first $d$
positions, and fill all other positions with zero. There are $d^k$
distinct strings. Consecutive ones are at distance at least $d+1$, so
every length-$m$ substring with $m<d$ contains at most one $1$. Hence
\[
 D_m\le m+1\le2m^s\quad(m<d),\qquad
 D_m\le n\le d^s\le m^s\quad(m\ge d).
\]
The family lies in $\mathcal F_{n,s,2}\subseteq\mathcal F_{n,s,A}$, and
Kraft's inequality forces every prefix code to have length at least
$\log(d^k)$ on some member. For fixed $s>1$,
\[
 k\log d\sim\frac1{2s}n^{1-1/s}\log n.
\]

For $s=1$, use the binary prefix families of Kociumaka, Navarro, and
Prezza~\cite[Definition IV.2, Lemma IV.3, and the binary construction
in the proof of Lemma IV.4]{WZIPKociumaka2023}. Their members satisfy
$D_m\le2m$, and the logarithm of their number is $\Omega((\log n)^2)$.
They therefore lie in $\mathcal F_{n,1,A}$, and the same counting argument
proves the lower bound; this converse for linear complexity is theirs.
\hfill$\square$

Single-window LZ77 attains these orders as well; the theorem shows that
they are optimal when the window is comparable to the input length. The
hard strings may depend on $n$. The next corollary shows that the windows
do not destroy this optimality.

\begin{corollary}[Preservation across window schedules]
\label{cor:schedule-minimax}
Fix $A\ge2$ and $s\ge1$. Consider a collection of complete prefix-coded
exact greedy WLZ families satisfying Theorem~\ref{thm:block-transfer},
with common bounds on literal and logarithmic length costs,
$t_B\le c_tB$, and $cn\le W=2^B\le n$ for common $c_t,c>0$.
Schedules are shared or transmitted, and all header costs are
$O((\log n)^2)$, uniformly over inputs and the collection.
Every such family satisfies
\[
 \max_{x\in\mathcal F_{n,s,A}} C_{W,n}(x)=\Theta(L^*_{n,s,A}),
\]
with uniform constants over the collection and the orders
in~\eqref{eq:finite-minimax}.
\end{corollary}
\noindent\emph{Proof.}
Corollary~\ref{cor:finite-polynomial} and the header bound give the
upper orders. Each complete encoder is one of the prefix codes in
Theorem~\ref{thm:finite-minimax}, so its converse gives the lower orders.
\hfill$\square$

For example, the literal/gamma format in the proof of the theorem allows
arbitrary integer widths $w_B(\ell)\in\{0,\ldots,B\}$ for
$2\le\ell<t_B$, below a threshold $t_B\ge2$, and width $B$ thereafter.
For $t_B=O(B)$, a self-delimiting threshold followed by its table of short
widths costs $O(B\log(B+1))$ bits. Substantial restrictions of short
matches, nonmonotone tables included, thus keep the optimal orders even
when their descriptions are transmitted; the length-two schedule is the
simplest example.

If the fitted tables are given to the decoder for free, as in
Proposition~\ref{prop:field-entropy}, the same comparison bounds the
payload more sharply. Here every selected conceptual match is coded as a
copy, without literal fallback.

\begin{corollary}[Linear-complexity Huffman payload]
\label{cor:linear-entropy-payload}
Use the exact greedy schedule from Theorem~\ref{thm:finite-minimax},
with $W=2^{\lfloor\log n\rfloor}$, $W_B(2)=1$, and
$W_B(\ell)=W$ for $\ell\ge3$. Encode every selected conceptual
match as a copy and every unmatched symbol as a literal, using the
four Huffman streams. For binary inputs satisfying
$D_m(x_1^n)\le Am$ for every $1\le m\le n$, with fixed $A\ge1$,
the payload obeys
\begin{equation}
 C_{\rm H}^{\rm pay}(x_1^n)=O(\log n\log\log n).
 \label{eq:linear-entropy-payload}
\end{equation}
Adding $O(\log n)$ stream framing preserves this order.
\end{corollary}
\noindent\emph{Proof.}
The counting of starts in the proof of Theorem~\ref{thm:block-transfer}
gives $P\le t_BZ_W$ for the number of conceptual phrases: at most $t_B-1$
starts other than the last lie in the final $t_B-1$ positions of each
reference copy, at most one further start is the last, and a reference
literal contains at most one start. This also counts long phrases and
phrases that cross reference boundaries. Here $t_B=3$, and the at most two
reference blocks have $Z_W=O(\log n)$ phrases in total
by~\eqref{eq:known-delta}. Thus $P\le3Z_W=O(\log n)$,
and~\eqref{eq:four-field-support} proves
\eqref{eq:linear-entropy-payload}. \hfill$\square$

The parse and its phrase boundaries are unchanged; only their
representation changes under the fitted codes. The result improves the
logarithmic factor of fixed-width fields for this schedule, not the
polynomial exponent: for fixed $s>1$ and $D_m(x_1^n)\le Am^s$ for every
$1\le m\le n$, the same reasoning gives $P=O(n^{1-1/s})$ and keeps the
order $O(n^{1-1/s}\log n)$. An arbitrary threshold $t_B=O(B)$ can add a
factor to $P\le t_BZ_W$, so the refinement does not hold uniformly over
the schedules of Corollary~\ref{cor:schedule-minimax}.

The free tables hide a cost, which the existing converse quantifies. Let
$X$ be uniform on the binary family used in the proof of
Theorem~\ref{thm:finite-minimax} for $s=1$, so that
$\mathsf H(X)=\Omega((\log n)^2)$ and $D_m\le2m$
\cite[Definition IV.2 and Lemmas IV.3--IV.4]{WZIPKociumaka2023}. With
deterministic parsing and Huffman tie rules, let $T\eqdef T(X)$ be the
labeled tables, and let $Q(X)$ be the payload with its $O(\log n)$
framing. For each fixed $T$, these framed payloads form a conditional
prefix code that recovers $X$. Therefore
\begin{equation}
 \mathsf H(T)=\mathsf H(X)-\mathsf H(X\mid T)
 \ge\mathsf H(X)-\mathbb E|Q(X)|
 =\Omega((\log n)^2).
 \label{eq:huffman-table-information}
\end{equation}
Any standalone prefix description of the tables has at least this
expected length on the family. Sending the $O(\log n)$ framing bits first
does not change this order, since conditioning on them removes at most
$O(\log n)$ bits of entropy. The bound holds on average over the family,
hence for some input in it, not for every input.

Conversely, $P=O(\log n)$ bounds the support sizes, and every integer
label needs $O(\log n)$ bits, so~\eqref{eq:huffman-table-cost} gives a
complete length of $O((\log n)^2)$. For fixed $A\ge2$, the complete
four-field format therefore attains the worst-case order
$\Theta((\log n)^2)$ on $\mathcal F_{n,1,A}$: its smaller payload moves the
distinguishing information into the labeled tables while keeping the
minimax order of the complete code.

\subsection{Universality}
\label{sec:greedy-universality}
The comparison with a reference parse also extends the classical
universality of sliding-window LZ~\cite{WZIPWynerZiv1994} to WLZ. The
proof divides the input into blocks of length $m$ and pays for each block
either about one phrase, if its content recurs within the window, or its
literals otherwise. This avoids any quantitative bound on recurrence
times at the starts chosen by the greedy parser.

\begin{lemma}[Cost of matched and unmatched blocks]
\label{lem:greedy-block-charge}
Retain the exact greedy model and uniform code-cost bounds above, and put
$A_0\eqdef\max\{b,a_0\}$. Partition $x_1^n$ into $J$ nonempty consecutive
blocks of lengths at most $m$, where $t_B\le m\le W$. Designate as
\emph{good} any set of full length-$m$ blocks that each have a copy
within distance $W$ starting inside the input, and call all other blocks
\emph{bad}. Let $J_g$ count the good blocks and $V$ the symbols in bad
blocks. Then
\begin{equation}
 C_B^{\rm greedy}(x_1^n)-\Gamma_n
 \le J\left[B+A_0+a_1\log\left(1+\frac nJ\right)\right]
       +b(t_B-1)J_g+bV.
 \label{eq:greedy-block-charge}
\end{equation}
A match cap of $W$ suffices.
\end{lemma}

\noindent\emph{Proof.}
Assign each conceptual phrase to the block containing its start, and
distinguish the last start in each block that contains one. In a good
block, any start with at least $t_B$ symbols remaining has a legal copy
of the entire remaining suffix at the witness distance, so greedy parsing
reaches the end of the block. Every other start lies in the final
$t_B-1$ symbols, and its phrase ends before another start in that block;
by literal fallback these disjoint intervals cost at most $b(t_B-1)$ in
total. In a bad block, the phrases other than the last likewise stay
inside the block and cost at most $b$ times its length.

The phrases of the distinguished last starts are disjoint even when they
cross block boundaries. There are at most $J$ of them, their lengths sum
to at most $n$, and each costs at most $B+A_0+a_1\log(\ell+1)$, a literal
included. Padding their list of lengths with zeros to $J$ entries and
using the concavity of $\log(1+u)$ bounds their total cost by the first
term in~\eqref{eq:greedy-block-charge}. Summing these costs proves the
result; the blocks serve only the analysis and do not reset the encoder.
\hfill$\square$

\begin{theorem}[Universality of exact greedy WLZ]
\label{thm:greedy-universality}
Suppose $t_B=o(B)$, every length from $t_B$ through $W=2^B$ has the full
window, and the exact greedy family satisfies the preceding uniform
cost bounds. Assume complete prefix coding with $\Gamma_n=o(n)$ for
each fixed $B$. For every finite-alphabet stationary ergodic source of
entropy rate $h\ge0$, put
$\overline K_m\eqdef B+A_0+a_1\log(m+1)$. For fixed
$t_B\le m\le W$,
\begin{equation}
 \limsup_{n\to\infty}\frac{C_B^{\rm greedy}(X_1^n)}n
 \le\frac{\overline K_m+b(t_B-1)}m+b\Pr\{R_m>W\}
 \quad\text{almost surely},
 \label{eq:greedy-recurrence}
\end{equation}
where $R_m$ is the backward recurrence distance in a two-sided stationary
extension. Consequently,
\begin{equation}
 \lim_{B\to\infty}\limsup_{n\to\infty}
       \frac{C_B^{\rm greedy}(X_1^n)}n=h
 \quad\text{almost surely}.
 \label{eq:greedy-universality}
\end{equation}
Both conclusions hold with a match cap of $W$.
\end{theorem}

\noindent\emph{Proof.}
For each of the $m$ phases, partition the input into length-$m$ blocks
with shorter end blocks. Mark the end blocks and the full blocks starting
at positions at most $W$ as bad, and call every other block good exactly
when $R_m\le W$. Across all phases, each interior position is a block
start once, so the average fraction of symbols in bad blocks is at most
\[
 \frac1n\sum_{i=W+1}^{n-m+1}
       \mathbf1\{R_m(T^{i-1}X)>W\}+O((W+m)/n),
\]
where $T$ is the one-symbol shift. Ergodicity bounds the upper limit of
this average by $\Pr\{R_m>W\}$; no ergodicity under the $m$-symbol shift
is needed. Apply Lemma~\ref{lem:greedy-block-charge} to a phase with the
fewest symbols in bad blocks. For every phase, $J/n\to1/m$, $n/J\to m$,
and $J_g\le J$, which proves \eqref{eq:greedy-recurrence}.

The recurrence theorem gives $\log R_m/m\to h$ almost surely
\cite{WZIPOrnsteinWeiss1993,WZIPJainBansal2014}. For any $\epsilon>0$,
choose $m_B\eqdef\lfloor B/(h+\epsilon)\rfloor$, which eventually lies
between $t_B$ and $W$. The recurrence theorem and bounded convergence
give $\Pr\{R_{m_B}>2^B\}\to0$. Also
$\overline K_{m_B}/m_B\to h+\epsilon$ and $t_B/m_B\to0$.
Thus the outer upper limit in~\eqref{eq:greedy-universality} is at most
$h+\epsilon$, also when $h=0$.

For the converse, write $I_n\eqdef-\log P_n(X_1^n)$. For every fixed $B$,
Kraft's inequality gives
\[
 \Pr\{C_B^{\rm greedy}(X_1^n)<I_n-2\log n\}\le n^{-2}.
\]
The Borel--Cantelli lemma and the Shannon--McMillan--Breiman theorem
yield $\liminf_n C_B^{\rm greedy}(X_1^n)/n\ge h$ almost surely.
Intersecting the probability-one events over the countably many $B$,
$m$, and positive rational $\epsilon$, then letting $\epsilon\downarrow0$,
proves~\eqref{eq:greedy-universality}.\hfill$\square$

Minimum-bit WLZ in the same format costs no more than exact greedy WLZ,
so it inherits universality:
\begin{equation}
 \lim_{B\to\infty}\limsup_{n\to\infty}
       \frac{C_B^*(X_1^n)}n=h\quad\text{almost surely}.
 \label{eq:ideal-universality}
\end{equation}
For this ideal parser, matched $m$-blocks and literal-coded failures give
the slightly sharper bound, averaged over phases,
\begin{equation}
 \limsup_{n\to\infty}\frac{C_B^*(X_1^n)}n
 \le\frac{B+a(m)}m+b\Pr\{R_m>W\}.
 \label{eq:ideal-recurrence}
\end{equation}
The greedy result needs neither unbounded match lengths nor a rate at
which the recurrence probability vanishes.

Universality concerns the limit; how fast the redundancy vanishes depends
on the length code, not on the windows. Appendix~\ref{app:obstruction}
proves that exact greedy parsing with fixed-width long pointers and
separately self-delimiting lengths has expected redundancy at least a
constant times $\log B/B$ on a uniform iid source, whatever the windows
of a fixed number of short-match classes, and
Section~\ref{sec:strict-comparisons} shows how a different length code
reduces it. The obstruction does not cover adaptive Huffman length codes
or a raw-block fallback coded at the entropy; earlier lower bounds for
sliding-window LZ likewise depend on the coding
model~\cite{WZIPLastras2006}. The modified LZ scheme of Wyner and Wyner
attains $O(1/B)$ under finite-memory
assumptions~\cite{WZIPWynerWyner1995}; removing a factor $\log B$ is a
different question from achieving a power $B^{-p}$ with $p>1$.

\section{What the Windows Gain, and What Sets the Leading Term}
\label{sec:strict-comparisons}
This section compares actual code lengths on the same source, not upper
bounds. Subsection~\ref{sec:strict-window-gain} shows that a calibrated
schedule never loses against a single window and gains a polynomially
small rate. Subsection~\ref{subsec:fixed-field-comparison} shows that the
leading redundancy term is set by the length code: a short codeword for a
growing longest match lowers it below that of the original fixed-field
LZ77 format, whatever cap that format uses.
Subsection~\ref{sec:run-escape} adds a separation from coding runs of one
symbol by their count. The bounds of Section~\ref{sec:convergence} do not
give these strict inequalities, which need separately specified encoder
families.

\subsection{What the windows gain}
\label{sec:strict-window-gain}
A calibrated length-dependent schedule can lower the actual asymptotic
rate of a single-window code. The baseline has literal cost $b>0$,
maximum distance $W=2^B$, and match cost $B+a(\ell)$, with $a(\ell)\ge0$;
both formats allow overlapping copies. Put
$L_B^{\mathrm{cut}}\eqdef\lfloor B/b\rfloor$ and use the schedule
\begin{equation}
 w_B(\ell)\eqdef\begin{cases}
 0,&\ell\le L_B^{\mathrm{cut}},\\
 B,&\ell>L_B^{\mathrm{cut}}.
 \end{cases}
 \label{eq:safe-window-schedule}
\end{equation}
Short matches thus have distance one, while long matches keep the entire
history. Supported lengths start at a fixed minimum and are unbounded.
Write $C_0^*$ and $C_W^*$ for the minimum token costs in the two formats;
headers of $o(n)$ bits for fixed $B$ do not affect the rates below.

\begin{proposition}[Pointwise comparison]
\label{prop:no-loss-windows}
For every input $x$, schedule~\eqref{eq:safe-window-schedule} satisfies
$C_W^*(x)\le C_0^*(x)$.
\end{proposition}
\noindent\emph{Proof.}
Replace every baseline match of length $\ell\le L_B^{\mathrm{cut}}$ by its
literals. This cannot increase the cost, because $B+a(\ell)\ge B\ge b\ell$.
All remaining matches keep their full window and original cost in the
new format, and the decoded history is unchanged. \hfill$\square$

The schedule thus never loses. It gains strictly when the input contains
short repeats that the baseline cannot afford, as on the following fixed
iid source.

\begin{theorem}[Strict saving with a polynomial lower bound]
\label{thm:strict-polynomial-saving}
Assume $0\le a(\ell)\le a_0+a_1\log_2(\ell+1)$ for fixed
$a_0,a_1\ge0$. For any $\beta>2$, take a binary iid source with
$\Pr\{X_i=1\}=p=2^{-b\beta}$, with $b$ large enough to encode a
binary literal and its discriminator. Under
schedule~\eqref{eq:safe-window-schedule}, almost surely,
\begin{equation}
 \Rrate_{W,B}^*\le \Rrate_{0,B}^*
       -\Omega\bigl((\log W)W^{-\beta}\bigr),
 \label{eq:strict-polynomial-gap}
\end{equation}
where $\Rrate_{j,B}^*\eqdef\limsup_{n\to\infty}C_j^*(X_1^n)/n$.
The implicit positive constant is independent of $B$.
\end{theorem}

\noindent\emph{Proof.}
Fix $0<\epsilon<\beta/2-1$ and define
\[
 r_B\eqdef\left\lceil\frac{(1+\epsilon)B}{b\beta}\right\rceil,
 \qquad k_B\eqdef L_B^{\mathrm{cut}}-2r_B.
\]
For large $B$, $k_B$ is a supported positive length. Call a run starting
at $i$ good if the preceding $W$ symbols contain no $1^{r_B}$,
$X_{i-1}=0$, $X_i^{i+L_B^{\mathrm{cut}}-1}=1^{L_B^{\mathrm{cut}}}$, and
$X_{i+L_B^{\mathrm{cut}}}=0$, and consider its middle interval
$[i+r_B,i+L_B^{\mathrm{cut}}-r_B-1]$, of length $k_B$.

No baseline match longer than $L_B^{\mathrm{cut}}$ can cover a middle
position, since such a match must extend beyond at least one end of the
run. If it extends left, copying the first $r_B$ ones of the run would
require an occurrence starting at $i-d\ge i-W$. For $d<r_B$, this
occurrence crosses the zero at $i-1$; for $d\ge r_B$, it is a forbidden
earlier $1^{r_B}$ wholly in the preceding window. If the match extends
right, it copies the last $r_B$ ones and the following zero. For
$d\le L_B^{\mathrm{cut}}$, that zero has source position
$i+L_B^{\mathrm{cut}}-d$ inside the current run, a contradiction. For
$d>L_B^{\mathrm{cut}}$, the source of the last $r_B$ ones lies wholly in
the preceding window, again a forbidden $1^{r_B}$. The argument uses
source positions at the run boundaries, so it also covers phrases
starting before the preceding window and overlapping copies.

Normalize any baseline parse by replacing all matches of length at most
$L_B^{\mathrm{cut}}$ by literals, as in
Proposition~\ref{prop:no-loss-windows}. Every middle interval is then
literal. Replace it by a match of length $k_B$ at distance one, which
saves
\[
 s_B\eqdef bk_B-a(k_B)
 =\left(1-\frac{2(1+\epsilon)}\beta\right)B-O(\log B)
 =\Theta(B)>0.
\]
Good runs, and hence the replaced intervals, are disjoint. Therefore, if
$N_B(x_1^n)$ counts the good runs fully contained in the input with
complete preceding windows,
\begin{equation}
 C_W^*(x_1^n)\le C_0(P)-s_BN_B(x_1^n)
 \label{eq:strict-finite-gap}
\end{equation}
for every baseline parse $P$, including a minimum-cost one.

Condition on the preceding zero. A union bound on the remaining iid
history shows that a good run has probability
\[
 \pi_B\ge (1-p)^2p^{L_B^{\mathrm{cut}}}(1-Wp^{r_B})
 \ge(1-p)^2W^{-\beta}(1-W^{-\epsilon}).
\]
For each fixed $B$, the ergodic theorem gives $N_B(X_1^n)/n\to\pi_B$
almost surely. Intersect these probability-one events over integer $B$,
divide~\eqref{eq:strict-finite-gap} by $n$, and take upper limits. Since
$s_B=\Theta(B)$, this proves~\eqref{eq:strict-polynomial-gap}.
\hfill$\square$

For a finite file, the explicit transformation in the proof saves at
least $s_BN_B(x_1^n)-\Delta\Gamma$ bits, headers included, with
$\Delta\Gamma$ as in~\eqref{eq:wzip-gain-condition}. For a heuristic WLZ
parse the saving is at least $s_BN_B(x_1^n)-E_W-\Delta\Gamma$, with $E_W$
from the same equation; either bound proves a strict improvement when it
is positive, and greedy parsing alone does not control $E_W$.

For example, $b=2$, $\beta=3$, and $p=1/64$ give a fixed binary source
with a saving of $\Omega((\log W)/W^3)$. This is a polynomial correction
to the single-window rate; it does not change the leading redundancy
term, of order $1/\log W$. The short classes in
\eqref{eq:safe-window-schedule} grow with $B$, unlike a fixed window table.
Minimum-bit parsing attains the theorem, and the explicit transformation
of a baseline parse attains~\eqref{eq:strict-finite-gap} without computing
either optimum. Any fixed short-distance width $w_0$ may replace zero,
reducing each run saving by only $w_0$ bits, so distance-one windows are
a convenient extreme rather than a requirement. The baseline here has
fixed-width distances; an entropy-coded single-window format is a
different comparator.

\subsection{The length code sets the leading term}
\label{subsec:fixed-field-comparison}
Appendix~\ref{app:obstruction} shows that with logarithmic-cost length
codes, exact greedy parsing has redundancy of order at least $\log B/B$ on
uniform iid data, whatever windows a fixed number of short-match classes
use. The leading term is therefore a property of the length code. This
subsection shows that a short codeword for a growing longest match lowers
it to $O(\log\log B/B)$, while the fixed-field format of the original
LZ77 paper, which copies a prefix and appends one source symbol with
fixed-width pointer and length fields~\cite[p.~338]{LZ77}, stays at
$\Omega(\log B/B)$ whatever its cap. The cap is of order $\log W$, a
different regime from the families with unbounded lengths of
Section~\ref{sec:convergence}. A single-window encoder with the same token
code obtains the same gain; the windows play no role in it.

\begin{lemma}[Uniform probability of finding a word]
\label{lem:uniform-word-hit}
In a uniform iid source on $q\ge2$ symbols, the probability that a
specified word of length $L\ge1$ does not occur in a block of $W$ symbols
is at most
\begin{equation}
 \epsilon(W,L)\eqdef2^{-\lfloor W/(8q^L)\rfloor}
 \le 2\exp\{- (\ln2)Wq^{-L}/8\}.
 \label{eq:uniform-word-hit}
\end{equation}
The bound holds uniformly over the specified word.
\end{lemma}

\noindent\emph{Proof.}
First take a block of $M=8q^L$ symbols, and let $Z$ count the occurrences
of the word at its $N\eqdef M-L+1$ possible starting positions. Its mean
is $\mu\eqdef Nq^{-L}\ge7$. Indicators at distance at least $L$ are
independent; at distance $1\le d<L$, their joint probability is either
zero or $q^{-(L+d)}$, according to whether the word permits that overlap.
Therefore
\[
 \operatorname{Var} Z
 \le \mu+2Nq^{-L}\sum_{d=1}^{L-1}q^{-d}
 \le \mu\left(1+\frac{2}{q-1}\right)\le3\mu.
\]
Chebyshev's inequality gives $\Pr\{Z=0\}\le3/\mu\le3/7<1/2$. Partition
the $W$ symbols into $\lfloor W/M\rfloor$ independent blocks of size $M$.
Absence from the whole window implies absence from every block, which
proves the first bound; the second follows by rounding. \hfill$\square$

\begin{theorem}[Smaller-order redundancy at a common phrase cap]
\label{thm:fixed-field-separation}
Let the source be uniform iid on $q=2^h$ symbols, where $h\ge1$ is an
integer. Let $W=2^B$, with $B$ a multiple of $h$, and give both encoders
the same history size $W$ and maximum phrase length
\begin{equation}
 L_B^{\mathrm{cap}}\eqdef\left\lfloor\frac{B-2\log\log B}{h}\right\rfloor.
 \label{eq:improved-phrase-cap}
\end{equation}
The baseline uses the original fixed-field triple format. WLZ uses exact
greedy parsing, supports all lengths from a fixed threshold $t$ through
$L_B^{\mathrm{cap}}$ in the full window, and may use smaller windows below
$t$. A WLZ cap token costs exactly $B+k_0$ bits for a constant $k_0$, and
all other token costs lie between zero and
$K_B=O(B+\log L_B^{\mathrm{cap}})$. Initialization and header costs are
$o(n)$ for fixed $B$. Then almost surely, for either the lower or the
upper large-file rate,
\begin{align}
 \Rrate_{77}(B)
 &=h+h\frac{\log B}{B}
       +2h\frac{\log\log B}{B}+O(B^{-1}),
 \label{eq:original-cap-expansion}\\
 \Rrate_W(B)
 &=h+2h\frac{\log\log B}{B}+O(B^{-1}).
 \label{eq:wzip-cap-expansion}
\end{align}
In particular, the comparison holds directly on the same input:
\begin{equation}
 \liminf_{n\to\infty}\frac{C_{77,B}(X_1^n)-C_{W,B}(X_1^n)}n
 \ge h\frac{\log B}{B}-O(B^{-1}).
 \label{eq:fixed-field-direct-gain}
\end{equation}
\end{theorem}

\noindent\emph{Proof.}
Write $L\eqdef L_B^{\mathrm{cap}}$, and call a source position bad if its
next $L$ symbols have no copy wholly contained in the preceding $W$
symbols. Conditional on these next symbols, the past is independent and
uniform, so Lemma~\ref{lem:uniform-word-hit} bounds the probability that a
position is bad by $\epsilon_B\eqdef\epsilon(W,L)$, and the ergodic
theorem bounds the almost-sure density of bad positions by $\epsilon_B$,
simultaneously for the countably many permitted $B$. Counting all bad
source positions bounds the number of bad phrase starts, so no
independence of the greedy phrase boundaries is needed.

At a start that is not bad and lies away from the end of the file, WLZ
emits a cap token. The original encoder also emits a phrase of the
maximum length, since it only needs to copy its first $L-1$ symbols. The
original codeword costs exactly
\[
 F_B\eqdef B+h\lceil\log_q L\rceil+h=B+\log L+O(1)
\]
bits, keeping the radix-$q$ fields of~\cite{LZ77}; for $q=2$ the cost is
$B+\lceil\log L\rceil+1$. If $N_{\rm bad}$ counts bad source positions,
each parser has at most $n/L+N_{\rm bad}+O(1)$ phrases. The original has
at least $n/L-O(1)$ phrases, and the number of WLZ cap tokens is between
$n/L-N_{\rm bad}-O(1)$ and $n/L$. Endpoint and initialization terms vanish
for fixed $B$. Thus almost surely
\begin{align*}
 \frac{F_B}{L}
 &\le\liminf_n\frac{C_{77,B}}n
 \le\limsup_n\frac{C_{77,B}}n
 \le\frac{F_B}{L}+F_B\epsilon_B,\\
 \frac{B+k_0}{L}-(B+k_0)\epsilon_B
 &\le\liminf_n\frac{C_{W,B}}n
 \le\limsup_n\frac{C_{W,B}}n
 \le\frac{B+k_0}{L}+K_B\epsilon_B.
\end{align*}
Now $Wq^{-L}\ge(\log B)^2$, so $B\epsilon_B$ decays faster than every
inverse power of $B$. Subtracting the WLZ upper bound from the baseline
lower bound, and expanding $L$, proves all the assertions.
\hspace*{\fill}$\square$

The cost condition has an explicit prefix-code realization. Give the cap
match token the codeword $0$, and code each other literal or match-length
symbol by $1$ followed by a fixed-width index over at most
$q+L_B^{\mathrm{cap}}$ alternatives, appending the distance field to a
match. Then $k_0=1$ and $K_B\le B+1+\lceil\log(q+L_B^{\mathrm{cap}})\rceil$;
constant format tags change only $k_0$. A Huffman code over the same token
alphabet is no longer in total than this code, so with $o(n)$ model
overhead and at most $B$ bits per distance it inherits the WLZ upper bound
and the strict gain~\eqref{eq:fixed-field-direct-gain}. The equality
\eqref{eq:wzip-cap-expansion} concerns the specified token costs; a better
entropy coder may compress further. The cap symbol must have no extra
length field whose cost grows with $B$. DEFLATE already gives its maximum
length 258 a distinct symbol with no extra length
bits~\cite[Section 3.2.5]{DEFLATE}; the theorem requires a growing cap and
the stated cost bound, not DEFLATE's fixed limits or an arbitrary choice
of Huffman tables.

The comparison extends, again almost surely, to a baseline that chooses
its own cap.

\begin{proposition}[Original-format lower bound for every cap]
\label{prop:original-any-cap}
Under the same iid source and history size, let the original
fixed-field triple encoder choose any finite phrase cap $J_B\ge1$.
Use exact full-window greedy copying followed by its mandatory final
symbol, with no entropy-coded raw-block fallback. Uniformly over these
cap choices, almost surely
\begin{equation}
 \liminf_{n\to\infty}\frac{C_{77,B,J_B}(X_1^n)}n
 \ge h+h\frac{\log B}{B}-O(B^{-1}),
 \label{eq:original-any-cap}
\end{equation}
and the same bound holds for $\liminf_n\mathbb E C_{77,B,J_B}(X_1^n)/n$.
Consequently the WLZ rate in~\eqref{eq:wzip-cap-expansion}, which also
holds in expectation, has smaller-order redundancy than this original
format, almost surely and in expectation, even when the baseline chooses
its own cap.
\end{proposition}

\noindent\emph{Proof.}
At the start of an original phrase, the unrevealed input is fresh iid:
the previous phrase consumed its first mismatching symbol or stopped at
its cap. Formally, the position just before each phrase start is a
stopping time for the natural filtration of prefixes. Conditional on that
prefix, a copy of length $r$ at any specified distance has probability
$q^{-r}$, overlapping copies included. A union bound over the $W$
distances, followed by summation of tails, bounds the conditional mean
phrase length by
\[
 \mu_B\eqdef B/h+1+(q-1)^{-1}.
\]
It is also bounded by $J_B$; put $m_B\eqdef\min\{J_B,\mu_B\}$. Applying
these conditional bounds to successive phrases that start in a finite
file gives $\mathbb E N\ge(n-O(W))/m_B$, up to a terminal term that
vanishes for fixed $B,J_B$. This uses conditional expectations, not
independence of phrase lengths.

For the almost-sure bound, parse the infinite sequence $X_1^\infty$ in the
same way. Let $L_j$ be the length of its $j$th phrase and $\mathcal F_j$
the $\sigma$-field of the input before that phrase. Then $1\le L_j\le J_B$
and $\mathbb E[L_j\mid\mathcal F_j]\le m_B$, so the partial sums of
$L_j-\mathbb E[L_j\mid\mathcal F_j]$ form a martingale with increments
bounded by $J_B$. The Azuma--Hoeffding inequality and the Borel--Cantelli
lemma give $\limsup_k k^{-1}\sum_{j\le k}L_j\le m_B$ almost surely. Each
phrase decision depends only on the history and the next $J_B$ symbols,
so the parse of $X_1^n$ agrees with this parse in its $k_n$ phrases that
start at or before $n-J_B$, which cover at least $n-J_B$ symbols. Hence
$\liminf_n N/n\ge\liminf_n k_n/n\ge1/m_B$ almost surely.

Each phrase costs $F_{B,J_B}\eqdef B+h\lceil\log_q J_B\rceil+h$ bits, so
the large-file rate is at least $F_{B,J_B}/m_B$, almost surely and in
expectation. If $J_B<B/(2h)$, this exceeds $2h$. Otherwise
$F_{B,J_B}\ge B+\log B-O(1)$, and division by $\mu_B$ gives
\eqref{eq:original-any-cap}, with constants independent of $J_B$. The WLZ
bound in expectation follows by counting bad positions in expectation in
the preceding proof.
\hfill$\square$

In terms of $W$, the WLZ redundancy is $O(\log\log\log W/\log W)$,
against $\Omega(\log\log W/\log W)$ for the original format on this
source. The separation is caused by the length code, not by the smaller
windows. Modified LZ constructions attain $O(1/B)$ redundancy under
finite-memory assumptions~\cite{WZIPWynerWyner1995}; the present result is
a direct comparison with the original format, not a claim of the best
redundancy order.

\subsection{Runs of one symbol}
\label{sec:run-escape}
Long runs of one symbol exceed any match cap of order $\log W$. This
subsection calibrates the cap, adds a token for such runs, and shows what
it gains.

The cap can be calibrated to the recurrence scale. At a fixed position of
a two-sided finite-alphabet stationary ergodic source with entropy rate
$h>0$, let $\Lambda_W$ be the longest match available at a backward
distance of at most $W$. Inverting the recurrence law
$\log_2 R_m/m\to h$ gives
\begin{equation}
 \frac{\Lambda_W}{\log_2 W}\longrightarrow\frac1h
 \quad\text{almost surely},
 \label{eq:ordinary-match-scale}
\end{equation}
the classical consequence of the Ornstein--Weiss theorem for match
lengths~\cite{WZIPOrnsteinWeiss1993,WZIPJainBansal2014}. Thus
$L_{\rm ord}(W)\asymp(\log_2 W)/h$ is a natural ordinary cap for a fixed
source of positive entropy, and the cap of
Theorem~\ref{thm:fixed-field-separation} approaches this scale from
below. The constant depends on the source, so no single cap is optimal
for all data. With $W=2^B$, a fixed-width length field over a range of
order $B/h$ needs $\log_2 B+O(1)$ bits for fixed $h$, whereas a
full-window distance field needs $B$ bits.

For uniform iid data on $q\ge2$ symbols, with $h=\log_2 q$, there is also
the elementary bound at a fixed position
\begin{equation}
 \Pr\{\Lambda_W\ge m\}\le\min\{1,Wq^{-m}\},\qquad m\ge1.
 \label{eq:ordinary-match-tail}
\end{equation}
At a fixed distance $d$, the match constraints leave $d$ free symbols
among $m+d$ iid symbols, so a match has probability $q^{-m}$, also with
overlap, and a union bound over the $W$ distances proves the bound.
Consequently a match of length $\lceil(1+\epsilon)B/h\rceil$ has
probability at most $W^{-\epsilon}$ at a fixed position, for any
$\epsilon>0$. This concerns a fixed position, not the longest match in a
file, and~\eqref{eq:ordinary-match-scale} does not apply to zero-entropy
data.

In the run-enabled format, the length symbols are
$1,2,3,\ldots,L_{\rm ord}$. Symbol $1$ announces a run: the pair $(r,1)$
is coded as $(1,r)$, meaning $r$ further copies of the last symbol at the
implicit distance one. Its second field is a variable-length count,
which may exceed $L_{\rm ord}$; the symbol $1$ is a selector, not a phrase
of length one. Symbols $\ell\ge2$ are ordinary match lengths, followed by
distances in the window $W(\ell)$, and all other repetitions use these
ordinary matches. A run of length $u\ge2$ that starts afresh needs a seed
literal and the count $u-1$; a single symbol needs only its literal. Let
$g(r)$ be the full cost of the selector and the count, with unambiguous
decoding.

A concrete prefix code uses the Elias gamma code $\gamma$:
\begin{equation}
 \begin{array}{ll}
 \text{binary literal }a: & 0a,\\
 \text{ordinary match }(\ell,d): &
       1\mid\gamma(\ell)\mid\operatorname{bin}_B(d-1),
       \quad 2\le\ell\le L_{\rm ord}(B),\\
 \text{run of }r\text{ additional copies}: &
       1\mid\gamma(1)\mid\gamma(r)=11\mid\gamma(r).
 \end{array}
 \label{eq:run-prefix-code}
\end{equation}
Since $\gamma(1)=1$ and $\gamma(\ell)$ starts with zero for $\ell\ge2$,
the ordinary code leaves the whole $11$ branch unused, so adding runs
keeps every ordinary codeword and its cost. Replacing a supported
distance-one match, $2\le r\le L_{\rm ord}(B)$, by the run token keeps its
expansion and saves exactly
\begin{equation}
 [B+1+|\gamma(r)|]-[2+|\gamma(r)|]=B-1
 \label{eq:run-swap-saving}
\end{equation}
bits, which is positive for $B>1$; counts beyond the ordinary cap give
the further benefit of longer expansions. The literal cost is $b=2$,
$a(\ell)=1+|\gamma(\ell)|$, and $g(r)=2+|\gamma(r)|\le3+2\log_2(r+1)$,
uniformly in $B$. The identity relies on this construction with
fixed-width distances: an arbitrary Huffman tree need not have unused
prefix space, refitting its codewords must pay for the changed costs, and
a byte-oriented implementation must be evaluated with its actual $g$.

\begin{proposition}[Coding by the number of runs]
\label{prop:run-count}
Let $x_1^n$ have $k$ maximal constant-symbol runs, of lengths
$u_1,\ldots,u_k$. With literal cost at most $\lambda$ and framing cost
$\Gamma_n$, there is a run parse of cost
\begin{equation}
 C_{\rm run}(x_1^n)\le \Gamma_n+\lambda k+
       \sum_{i:u_i\ge2}g(u_i-1).
 \label{eq:run-general-cost}
\end{equation}
If, for all $r\ge1$, $g(r)\le A+D\log_2(r+1)$ with constants
$A,D\ge0$, then
\begin{equation}
 C_{\rm run}(x_1^n)\le \Gamma_n+(\lambda+A)k+
                 Dk\log_2(n/k).
 \label{eq:run-count-cost}
\end{equation}
Both bounds also hold for the minimum-cost encoder admitting this parse.
\end{proposition}
\noindent\emph{Proof.}
Code each run by its seed and its remaining count; a single symbol needs
no count. Since $\sum_i u_i=n$, the concavity of the logarithm gives
$\sum_i\log_2 u_i\le k\log_2(n/k)$. \hfill$\square$

The proposition gives a further gain over the original fixed-field
format. Keep the phrase cap $L=L_B^{\mathrm{cap}}=\Theta(B)$ of
Theorem~\ref{thm:fixed-field-separation}, with $B=\log_2 W$. Each original
phrase costs $F_B=B+h\lceil\log_q L\rceil+h=\Theta(B)$ bits and emits at
most $L$ source symbols, its mandatory final symbol included. After at
most $O(W)$ initialization symbols, the number of phrases is at least
$(n-O(W))/L-O(1)$. Hence, when $W\to\infty$ and $W=o(n)$,
$C_{77,B}\ge F_B((n-O(W))/L-O(1))=\Omega(n)$, with all initialization and
framing costs included; this holds on every input, without an iid
assumption. With logarithmic count costs, $\Gamma_n=O(\log n)$, and
$k=O(n^\alpha)$ for fixed $0\le\alpha<1$, it follows that
\begin{equation}
 C_{77,B}(x_1^n)=\Omega(n),\qquad
 C_{\rm run}(x_1^n)=O(n^\alpha\log n)=o(n).
 \label{eq:run-original-boost}
\end{equation}
For example, $W=2^{h\lfloor\rho\log_2 n/h\rfloor}$ with fixed $0<\rho<1$
is an admissible joint scaling. The original cap includes its final
symbol, so the comparison counts phrase expansions consistently.

The run representation can be added to a complete ordinary code without
losing its guarantees. For complete prefix codes $C_0$ and $C_{\rm run}$ on
the same inputs, prefix their codewords by a one-bit selector and choose
the shorter branch. The resulting complete code satisfies
\begin{equation}
 C_{\rm aug}(x_1^n)
 \eqdef1+\min\{C_0(x_1^n),C_{\rm run}(x_1^n)\},
 \label{eq:run-augmentation}
\end{equation}
with all decoder parameters shared or included in the branch costs. With
$C_0=C_{77,B}$, this keeps the original code within one bit and
attains~\eqref{eq:run-original-boost} on inputs with few runs. With the
capped WLZ code as $C_0$, it keeps the uniform-iid bound of
Theorem~\ref{thm:fixed-field-separation} and adds this guarantee for few
runs. Unlike the periodic Huffman comparison, the run branch changes the
available phrase expansions, so its gain depends on the ordinary cap: the
original format can also code a constant input in $O(\log n)$ bits if it
may use a phrase cap as large as the input and a fixed window.

More generally, suppose every output-producing token of a baseline format
expands to at most $L_n$ symbols and costs at least a fixed $c>0$ bits.
Every baseline parse then satisfies
\begin{equation}
 C_{\rm cap}(x_1^n)\ge c\lceil n/L_n\rceil.
 \label{eq:run-cap-converse}
\end{equation}
Under logarithmic count costs, if $\Gamma_n=O(\log n)$, $k=O(n^\alpha)$,
and $L_n=O(n^\delta)$, with $\alpha,\delta\ge0$ and $\alpha+\delta<1$, then
\begin{equation}
 \frac{C_{\rm run}(x_1^n)}n
     =O\!\left(n^{-(1-\alpha)}\log n\right),\qquad
 \frac{C_{\rm cap}(x_1^n)}n=\Omega(n^{-\delta}).
 \label{eq:run-exponent-separation}
\end{equation}
The run code thus has a strictly larger decay exponent, up to the
logarithm. For logarithmic caps $L_n=O(\log W_n)$ with
$2\le W_n\le n^\eta$ and fixed $\eta>0$, \eqref{eq:run-cap-converse} gives
$\Omega(n/\log n)$ bits, more than the run code's $O(n^\alpha\log n)$ for
every fixed $\alpha<1$. On a constant input with a fixed cap $L$, the
corresponding rates are $O(\log n/n)$ and at least $c/L$. Unrestricted
overlapping LZ77 with a logarithmic length code also uses $O(\log n)$
bits on that input: a fixed history alone does not cap the expansion of a
match.

The rule for extending the count decides which conclusion applies.
Positional varints and universal integer codes satisfy the logarithmic
hypothesis, whereas one extension byte per additional 255 symbols, as in
the LZ4 match-length format~\cite{LZ4}, costs $\Theta(r)$ bits. More
generally, if $g(r)/r\to\gamma$ and $\Gamma_n=o(n)$, the parse with a
single run has limiting rate $\gamma$ on a constant input, which still
improves strictly on the capped format whenever $\gamma<c/L$; for one byte
per $K$ additional symbols with bounded other overhead, $\gamma=8/K$. A
fixed cap on run counts, or compulsory splitting into fixed-size blocks,
must be included in~\eqref{eq:run-general-cost}, and then the conclusion
\eqref{eq:run-exponent-separation} for unbounded counts does not hold.

The good runs in the proof of Theorem~\ref{thm:strict-polynomial-saving}
also give a run-token analogue on that source: if a decodable
augmentation keeps every baseline token at its original cost and admits
run tokens whose count cost is at most $\kappa r+O(\log(r+1))$ for a
fixed $0\le\kappa<b$, as~\eqref{eq:run-prefix-code} does with $\kappa=0$,
then coding each literal middle interval, which follows a one, as a run
saves $(b-\kappa)k_B-O(\log B)=\Theta(B)$ bits, so the minimum-cost rate
again falls by $\Omega((\log W)W^{-\beta})$ almost surely, under any
common ordinary cap. Like that theorem, this improves a polynomial
correction rather than the leading redundancy order, and the few-run
condition in~\eqref{eq:run-exponent-separation} describes finite-input
classes, since for a stationary ergodic source $k/n\to\Pr\{X_0\ne X_1\}$
and sublinear run counts force a constant sequence.

Dedicated run tokens have precedents. Pucrunch combines capped LZ77
matches with longer run-length tokens~\cite{WLZPucrunch}, and the Linux
LZO-RLE format reinterprets a reserved distance pattern as a zero-run
instruction, aimed at data with many zeros such as zram
workloads~\cite{WLZLZORLE}; Pucrunch also restricts two-byte matches to
shorter distances than longer matches. The results here quantify the
reserved-length variant under explicit cost and parsing assumptions.
WZIP~\cite{WZIPWuCodecs2026} implements it in WZIP\_L, the variant used for
inputs of 32\,KB or more~\cite{WZIPCode2026}: a reserved length symbol
announces a run of copies of the preceding byte, and the run's count, not
bound by the match cap, is coded in the offset field with that field's
Huffman code and raw extra bits. Runs leave the offset cache unchanged,
and consecutive runs merge into one token.

\section{Window Selection and Entropy Estimation}
\label{sec:finite-data}
The previous sections fix the window or let it grow. This section chooses
it from the data. Subsection~\ref{sec:window-objective} states the
objective for a given input, and Subsection~\ref{sec:entropy-estimation}
turns window selection into an entropy-rate estimator.

\subsection{Choosing the windows for a given input}
\label{sec:window-objective}
For a given input $x_1^n$, let $C_n(\boldsymbol W;x)$ be the complete
code length under a specified parser and format, where
$\boldsymbol W$ lists the windows of the length classes. The compression
objective under a memory budget $M$ is
\begin{equation}
 \min_{\boldsymbol W:\,\mathsf M(\boldsymbol W)\le M}
 C_n(\boldsymbol W;x),\qquad
 \Rrate_n\eqdef\frac{C_n}{n},\qquad \mathrm{CR}_n\eqdef\frac{b_{\rm raw}}{\Rrate_n}.
 \label{eq:finite-objective}
\end{equation}
Here $\Rrate_n$ is in bits per input symbol, $b_{\rm raw}$ is the raw
storage cost per symbol (8 for bytes), and $\mathrm{CR}_n$ is the raw size
divided by the compressed size. The memory includes the indexes of
Section~\ref{sec:wzip-theory}, and a time constraint can be added. The
code length $C_n$ includes headers, schedules, dictionaries not already
shared with the decoder, and entropy models. The best schedule generally
depends on the input and the coding model, not on the length $n$ alone.
Practical defaults, such as zlib's 32\,KiB history, the largest DEFLATE
permits, and Zstandard's parameters by level and input
size~\cite{WZIPZlibManual,DEFLATE,WZIPZstdAPI2025}, are not chosen to
minimize this objective for each input; adaptive choice of the LZ77
history also has theoretical precedents for nonstationary
sources~\cite{WZIPReifStorer2001}.

If enlarging a window only adds admissible matches and keeps every
existing token cost, the minimum-bit cost cannot increase, since the old
parse remains feasible. An optimal window strictly inside the range
therefore requires some other effect, such as distance fields that widen
with the window, model overhead, or a resource constraint; a heuristic
search budget can also change the parse. In particular, the fixed-width
cost $\log W$ and the variable distance costs of a modern format are
different optimization models.

Selecting among $K$ candidate codes known to the decoder costs little:
sending the index of the shortest in $\lceil\log K\rceil$ bits, followed
by its codeword, is within $\lceil\log K\rceil$ bits of the best candidate
on every input. This is standard two-part coding. For the grid
$W\in\{1,2,4,\ldots,2^{\lfloor\log n\rfloor}\}$, selection costs
$O(\log\log n)$ bits. For $J$ nondecreasing integer widths between $0$ and
$B$, there are $\binom{B+J}{J}$ schedules, so a shared enumeration costs
at most $\lceil\log\binom{B+J}{J}\rceil$ bits. The guarantee requires
evaluating the candidates on the whole input; selecting from a sample
needs a separate argument.

\subsection{Entropy estimation}
\label{sec:entropy-estimation}
The target is the entropy rate $h\eqdef\lim_{n\to\infty}\mathsf H(P_n)/n$
of an unknown finite-alphabet stationary ergodic source, observed through
one sequence $X_1^n$. Estimating it by compression is well
established~\cite{WZIPKontoyiannis1998,WZIPGao2008}. Here, window
selection turns the iterated limit of
Theorem~\ref{thm:greedy-universality}, first in $n$ and then in $B$, into
one consistent estimator along the observed sequence, with explicit
selection and finite-data costs.

Fix one candidate family with parsing convention
$\mathrm{enc}\in\{\mathrm{greedy},*\}$, meaning exact greedy or
minimum-bit parsing; Appendix~\ref{app:exact-parsing} bounds the time to
compute the latter. For each $B\ge1$, let $C^{\mathrm{enc}}_{B,n}(x_1^n)$
be its complete binary prefix-code length, with maximum history $2^B$ and
the hypotheses of Theorem~\ref{thm:greedy-universality}. The decoder knows
$n$ and the candidate family, and $C^{\mathrm{enc}}_{B,n}$ includes the
initialization and all other required metadata; the parsing convention is
fixed in advance and needs no selector bit. These are actual code lengths,
not idealized token costs. Let $b_{\rm raw}$ be the cost per symbol of a
valid fixed-width raw representation. For $n\ge2$, choose deterministic
integers
\[
 1\le K_n\le\lfloor\log n\rfloor,\qquad K_n\longrightarrow\infty,
 \qquad d_n\eqdef\lceil\log K_n\rceil;
\]
the simplest choice is $K_n=\lfloor\log n\rfloor$. Define
\begin{equation}
 \begin{split}
 \mathcal L_n(x_1^n)
  &\eqdef1+\min\left\{nb_{\rm raw},\;
       d_n+\min_{1\le B\le K_n}C^{\mathrm{enc}}_{B,n}(x_1^n)\right\},\\
 \widehat h_n&\eqdef\mathcal L_n(X_1^n)/n.
 \end{split}
 \label{eq:wlz-entropy-estimator}
\end{equation}
The code sends a one-bit raw/compressed flag; in the compressed branch it
sends a minimizing $B$ in $d_n$ bits, with a deterministic rule for ties,
followed by the WLZ codeword. The selected window in that branch is
$\widehat W_n\eqdef2^{\widehat B_n}$, and the raw branch guarantees
$\mathcal L_n\le nb_{\rm raw}+1$. If $n$ is not shared, its
self-delimiting description adds $O(\log(n+1))$ bits to all bounds below.

\begin{theorem}[Consistency under window selection]
\label{thm:wlz-entropy-consistency}
Assume the hypotheses of Theorem~\ref{thm:greedy-universality}:
all lengths from $t_B=o(B)$ through $W=2^B$ use the full window,
fixed match costs are at most $B+a_0+a_1\log(\ell+1)$ uniformly in $B$,
and complete initialization and header costs are $o(n)$ for every
fixed $B$. A cap at $W$ is permitted. Exact greedy candidates use the
literal-fallback rule.
Then, for every finite-alphabet stationary ergodic source,
\begin{equation}
 \widehat h_n\longrightarrow h\quad\text{almost surely},
 \qquad \mathbb E|\widehat h_n-h|\longrightarrow0.
 \label{eq:wlz-entropy-consistency}
\end{equation}
Selection costs at most $1+d_n=O(\log\log n)$ bits beyond the best
available WLZ candidate.
\end{theorem}

\noindent\emph{Proof.}
The flag and the fixed-width index make $\mathcal L_n$ a prefix code for
each known $n$. Every fixed $B$ is eventually among the candidates, so
\[
 \limsup_{n\to\infty}\widehat h_n
 \le\inf_{B\ge1}\limsup_{n\to\infty}\frac{C^{\mathrm{enc}}_{B,n}(X_1^n)}n
 \le h
 \quad\text{almost surely},
\]
by Theorem~\ref{thm:greedy-universality} and $d_n=o(n)$.
For the reverse inequality put $I_n\eqdef-\log P_n(X_1^n)$. Kraft's
inequality gives, for every $s>0$,
\begin{equation}
 \Pr\{\mathcal L_n(X_1^n)<I_n-s\}
 \le 2^{-s}\sum_{x_1^n}2^{-\mathcal L_n(x_1^n)}
 \le2^{-s}.
 \label{eq:wlz-kraft-tail}
\end{equation}
Taking $s=2\log n$, the Borel--Cantelli lemma and the
Shannon--McMillan--Breiman theorem give $\liminf_n\widehat h_n\ge h$
almost surely. Finally $0\le\widehat h_n\le b_{\rm raw}+1/n$, and
dominated convergence proves convergence in mean absolute
error.\hfill$\square$

\begin{corollary}[Finite-data redundancy and estimation error]
\label{cor:wlz-entropy-error}
The expected redundancies
$r_{B,n}\eqdef\mathbb E C^{\mathrm{enc}}_{B,n}/n-h$ and
$\rho_n\eqdef\mathbb E\widehat h_n-h$ satisfy
\begin{equation}
 0\le\rho_n\le\frac1n+
 \min\left\{b_{\rm raw}-h,\;
       \frac{d_n}{n}+\min_{1\le B\le K_n}r_{B,n}\right\}.
 \label{eq:wlz-entropy-oracle}
\end{equation}
Writing $u_+\eqdef\max\{u,0\}$, the absolute estimation error satisfies
\begin{equation}
 \mathbb E|\widehat h_n-h|
 \le \rho_n+2\mathbb E(h-I_n/n)_+ +\frac{2}{n\ln2}.
 \label{eq:wlz-entropy-error}
\end{equation}
In particular, for an iid source with single-symbol law $p$ and
$\sigma^2\eqdef\operatorname{Var}[-\log p(X_1)]$,
\begin{equation}
 \mathbb E|\widehat h_n-h|
 \le\rho_n+\frac{\sigma}{\sqrt n}+\frac{2}{n\ln2}.
 \label{eq:wlz-iid-entropy-error}
\end{equation}
For uniform iid data, $\sigma=0$.
\end{corollary}

\noindent\emph{Proof.}
The entropy bound for prefix codes gives
$\mathbb E\mathcal L_n\ge\mathsf H(P_n)\ge nh$ and likewise $r_{B,n}\ge0$.
Taking expectations in~\eqref{eq:wlz-entropy-estimator} and using
$\mathbb E\min_j Z_j\le\min_j\mathbb E Z_j$ proves
\eqref{eq:wlz-entropy-oracle}. Integrating~\eqref{eq:wlz-kraft-tail}
gives $\mathbb E(I_n-\mathcal L_n)_+\le1/\ln2$. Now use
\[
 \mathbb E|\widehat h_n-h|
 =\rho_n+2\mathbb E(h-\widehat h_n)_+,
 \qquad
 (h-\widehat h_n)_+
 \le(h-I_n/n)_++(I_n-\mathcal L_n)_+/n.
\]
For iid data, $\mathbb E(I_n/n)=h$ and
$\operatorname{Var}(I_n/n)=\sigma^2/n$, so
$2\mathbb E(h-I_n/n)_+=\mathbb E|I_n/n-h|\le\sigma/\sqrt n$.
\hspace*{\fill}$\square$

These bounds separate the coding redundancy from the fluctuation of the
source information. Any finite-$n$ bound proved for the candidate codes
carries over through~\eqref{eq:wlz-entropy-oracle}: the estimator inherits
the smallest such bound over $B\le K_n$, plus $(1+d_n)/n$. Other prefix
codes, such as the family with a cheap cap token, may be added to the
selector; keeping the universal WLZ candidates preserves consistency, the
index must then cover the enlarged list, and each redundancy bound applies
only to the format for which it was proved. The maximum history $2^{K_n}$
may grow arbitrarily slowly, since every fixed window eventually becomes
available, and the selected window itself need not grow. A claim about
finite-sample accuracy on a corpus requires control of both terms
in~\eqref{eq:wlz-entropy-error}; the size of the corpus alone controls
neither.

For a fixed nonuniform iid source, the bound~\eqref{eq:wlz-entropy-oracle}
can be made explicit without adding another compression method to the
selector, by an elementary occurrence and concentration argument that
covers both candidate families.

\begin{proposition}[Nonuniform iid sources]
\label{prop:nonuniform-iid}
Let $p$ be a fixed nonuniform iid law with full support on a finite
alphabet of at least two symbols. Put
\[
 h\eqdef-\sum_a p_a\log p_a>0,\qquad
 \Delta_p\eqdef\log(p_{\max}/p_{\min})>0.
\]
For either candidate family, put
$r^{\mathrm{enc}}_{B,n}\eqdef\mathbb E C^{\mathrm{enc}}_{B,n}/n-h$.
Let $W=2^B$, let literals cost at most $b$, and let matches of length
$m$ cost at most $B+a(m)$. The deterministic bound $\Gamma(n,B)$
includes all header and framing costs beyond the literals counted below.
For every $\epsilon>0$ and integer $t_B\le m\le W\le n$, assuming
that every distance through $W$ is allowed at length $m$, define
\[
 e(W,m,\epsilon)\eqdef
 \exp\!\left(-\frac{2m\epsilon^2}{\Delta_p^2}\right)
 +\exp\!\left(-\left\lfloor\frac Wm\right\rfloor
                         2^{-m(h+\epsilon)}\right).
\]
The minimum-bit candidate satisfies
\begin{equation}
 0\le r^*_{B,n}\le
 \frac{B+a(m)}m-h+b\,e(W,m,\epsilon)
 +\frac{b(W+m)+\Gamma(n,B)}n.
 \label{eq:nonuniform-iid-finite}
\end{equation}
For exact greedy candidates satisfying Lemma~\ref{lem:greedy-block-charge},
with $\overline K_m$ from Theorem~\ref{thm:greedy-universality},
\begin{equation}
 \begin{split}
 0\le r^{\mathrm{greedy}}_{B,n}
 &\le\frac{\overline K_m+b(t_B-1)}m-h+b\,e(W,m,\epsilon)\\
 &\quad+\frac{2\overline K_m+b(W+m)+\Gamma(n,B)}n.
 \end{split}
 \label{eq:nonuniform-iid-greedy-finite}
\end{equation}
For $n<W$, literal fallback instead gives
$r^{\mathrm{enc}}_{B,n}\le b-h+\Gamma(n,B)/n$.
If $t_B=o(B)$ and $a(m)=O(\log(m+1))$ uniformly in $B$, then,
uniformly for sufficiently large $B$ and $n\ge2^B$,
\begin{equation}
 r^*_{B,n}\le O\!\left(\sqrt{\frac{\ln B}{B}}\right)
       +O\!\left(\frac{2^B+B}{n}\right)+\frac{\Gamma(n,B)}n.
 \label{eq:nonuniform-iid-rate}
\end{equation}
For exact greedy candidates, the same bound has the additional term
$O(t_B/B)$. The constants may depend on the fixed law $p$ and code-cost
bounds. For the minimum-bit selector with $K_n=\lfloor\log n\rfloor$, if
$\Gamma(n,B_n)=O(\log n)$ at $B_n=\lfloor(\log n)/2\rfloor$, then
\begin{equation}
 \mathbb E|\widehat h_n-h|
 =O\!\left(\sqrt{\frac{\log\log n}{\log n}}\right).
 \label{eq:nonuniform-entropy-rate}
\end{equation}
For the exact greedy selector, the bound in~\eqref{eq:nonuniform-entropy-rate}
has the additional term $O(t_{B_n}/B_n)$. In particular, the same
square-root order holds when $t_B=O(\sqrt{B\ln B})$; the weaker
$t_B=o(B)$ still ensures consistency.
\end{proposition}

\noindent\emph{Proof.}
For a comparison parse, send the first $W$ symbols as literals, then
process consecutive blocks of length $m$ at fixed positions, and send the
final fewer than $m$ symbols as literals. At each full block, use one
match if its word occurs among $N\eqdef\lfloor W/m\rfloor$ disjoint
length-$m$ blocks in the preceding window, and literals otherwise. All
these matches are admissible. Fixed block starts, unlike starts selected
by a greedy parser, allow the iid law to be used directly, and no
independence between the success events at different blocks is needed.
The chosen history blocks end before the query begins and start at
distances between $m$ and $W$, so the comparison uses only nonoverlapping
copies, although the format allows overlap.

The query word is independent of its preceding window. Conditional on
query $u$, the selected past words are independent, each equal to $u$
with probability $P_m(u)$, so all of them differ from $u$ with probability
at most $(1-P_m(u))^N\le\exp(-NP_m(u))$. The self-information
$I_m\eqdef-\log P_m(X_1^m)$ of the query is a sum of independent variables
with mean $h$ and range length $\Delta_p$, and Hoeffding's inequality
\cite[Theorem 2]{WZIPHoeffding1963} gives
\[
 \Pr\{I_m>m(h+\epsilon)\}
 \le\exp(-2m\epsilon^2/\Delta_p^2).
\]
On the complementary event, $P_m(X_1^m)\ge2^{-m(h+\epsilon)}$. Splitting
the failure probability over these events proves the bound
$e(W,m,\epsilon)$; a bound exceeding one is harmless.

Write $n-W=Jm+r$, where $0\le r<m$. Each full block has expected cost at
most $B+a(m)+bm\,e(W,m,\epsilon)$, so the complete comparison code costs
at most
\[
 bW+J\{B+a(m)+bm\,e(W,m,\epsilon)\}+br+\Gamma(n,B)
\]
in expectation, and minimum-bit WLZ costs no more. Using $Jm\le n$ and
$r<m$ proves~\eqref{eq:nonuniform-iid-finite}; the lower bound follows
from the entropy bound for prefix codes. The cost of the first $W$
literals belongs to this comparison parse, not to WLZ itself.

For exact greedy parsing, split the first $W$ symbols into full $m$-blocks
and at most one shorter block, and call them bad. Partition the remaining
symbols into full $m$-blocks and a possibly shorter bad terminal block.
There are at most $n/m+2$ blocks, at most $n/m$ of them good, and the
expected number of symbols in bad blocks is at most
$W+m+n e(W,m,\epsilon)$. Apply Lemma~\ref{lem:greedy-block-charge}; since
every block has length at most $m$, its logarithmic term is at most
$\overline K_m$ per block. This proves
\eqref{eq:nonuniform-iid-greedy-finite} with no independence assumption
at greedy starts; again the initial window is part of the analysis, not a
raw start required of the encoder.

For the asymptotic choice, take
\[
 \epsilon_B=\Delta_p\sqrt{\frac{2h\ln B}{B}},\qquad
 m_B=\left\lfloor\frac{B-3\log B}{h+\epsilon_B}\right\rfloor.
\]
Since $m_B\sim B/h$, the full window eventually admits this length.
Also $hm_B/B\ge1/2$ for sufficiently large $B$, so the first failure term
is at most $B^{-2}$. For the second,
\[
 \left\lfloor\frac W{m_B}\right\rfloor
 2^{-m_B(h+\epsilon_B)}
 \ge\frac{B^3}{2m_B}=\Omega(B^2).
\]
The floor in $m_B$ and the length-code cost contribute $O((\log B)/B)$ to
$(B+a(m_B))/m_B-h$, leaving $O(\sqrt{(\ln B)/B})$, and substitution
proves \eqref{eq:nonuniform-iid-rate}. In the greedy bound,
$\overline K_{m_B}=B+O(\log B)$ and $b(t_B-1)/m_B=O(t_B/B)$, and the
remaining boundary terms are $O((2^B+B)/n)$.

Finally, the default candidate list contains $B_n$, with $2^{B_n}$ within
a factor two of $\sqrt n$. Its redundancy bound, the selector overhead
$O((\log\log n)/n)$, and the iid fluctuation term $O(n^{-1/2})$ in
Corollary~\ref{cor:wlz-entropy-error} imply
\eqref{eq:nonuniform-entropy-rate}, with the additional term
$O(t_{B_n}/B_n)$ for greedy candidates.\hfill$\square$

The length $m_B$ depends on the source but defines only a comparison
parse or an analysis partition; neither the candidate encoders nor the
selector need to know $p$, $h$, or $\Delta_p$. Every match already sends
its own length in the WLZ token code, so the comparison needs no separate
block-size message or probability model, and the literals and matches use
the candidate's own decoder. The rate requires the stated candidate list
and a header bound along $B=B_n$; headers that are $o(n)$ for each fixed
window, or a list that grows arbitrarily slowly, do not suffice. The
guarantee holds for each fixed nonuniform iid law; the uniform case
$\Delta_p=0$ is covered by the uniform-source analysis of
Section~\ref{sec:strict-comparisons}.

\section{Conclusion}
\label{sec:conclusion}
Match-length-dependent windows break the suffix closure on which the
classical analyses of LZ77 rely. The transfer theorem bounds the cost of
this loss, and of an inexact search, by the phrase count of a
single-window reference parse. With it, WLZ keeps the convergence rates
and the optimal finite-input rates of LZ77 when the full window starts at
length $O(\log W)$, and it is universal when the full window starts at
length $o(\log W)$. What the windows gain is real but small: a calibrated
schedule never loses and saves a polynomially small rate. The leading
redundancy term is set by the length code instead: on uniform iid data,
with logarithmic-cost lengths and exact greedy parsing, it cannot fall
below order $(\log\log W)/\log W$, whatever the short windows, and a short
codeword for a growing longest match lowers it below that of the original
LZ77 format. For codec design, short matches
below a full-window threshold of $o(\log W)$ can therefore be restricted
arbitrarily to save distance bits, search, and index memory, while the
leading compression gains must come from how lengths and distances are
coded. The measurements of WZIP and WLZ4~\cite{WZIPWuCodecs2026} agree with
this division: ablations indicate that length-conditioned offset coding
improves WZIP's compression, while the restricted windows reduce its
indexing resources.

A main theoretical question is whether the polynomial saving of
calibrated windows is sharp: matching upper and lower bounds for the
saving on a fixed source, with the length code held fixed, would settle it
and show when the windows could affect the leading term. Quantitative
results for dependent sources also call for bounds on recurrence tails and
search losses.

An algorithmic question is how to choose window sizes and hash-chain
inspection budgets jointly under memory and time constraints, in the
spirit of bicriteria parsing, which trades compressed size against
decompression time~\cite{WZIPBicriteria2013}. The goal is an approximation
guarantee for the complete code length, including model and selector
costs, without encoding every candidate. Combined with quantitative
source bounds, such a guarantee would extend the entropy-estimation
results to practical encoders.

\appendices
\section{A Positive-Entropy Redundancy Obstruction}
\label{app:obstruction}

This appendix proves the lower bound cited in
Subsections~\ref{sec:greedy-universality}
and~\ref{subsec:fixed-field-comparison}: with logarithmic-cost lengths,
no choice of windows for a fixed number of short-match classes reduces
the redundancy of exact greedy parsing below order $\log B/B$.

\begin{samepage}
\begin{theorem}[Fixed-threshold, separately coded lengths]
\label{thm:obstruction}
Let the input be iid uniform on $q\ge2$ symbols, with entropy
$h=\log q$. Let $W=2^B$. Suppose the exact greedy parser has a fixed
integer threshold $t\ge1$ independent of $B$, assigns every length
$\ell\ge t$ the full window and exactly $B$ distance bits, and supports
unbounded overlapping matches. Shorter supported lengths may use any
smaller windows. Suppose a long match costs $B+a(\ell)$, where
\[
 a(\ell)\ge c\log(\ell+1)-C,
 \qquad c>0,
\]
and all other costs are nonnegative. The initial $W$ symbols are emitted
before greedy parsing. Phrase boundaries are determined by exact
longest admissible matching, without future-dependent parse optimization.
Then
\begin{equation}
 \liminf_{n\to\infty}\frac{\mathbb E C_B(X_1^n)}n
 \ge h+ch\frac{\log B}{B}-O(B^{-1}).
 \label{eq:obstruction}
\end{equation}
The same conclusion permits replacing a conceptual matched phrase by
its literals at fixed cost $b\ell$ if $b>h$, provided its conceptual
greedy boundary is retained.
\end{theorem}
\end{samepage}

\noindent\emph{Proof.}
First bound the mean phrase length while accounting for lookahead.
Extend the file to an infinite iid sequence. Let $S_j$ and $L_j$ be the
successive conceptual phrase starts and lengths of the infinite parse,
and let $\mathcal F_i$ reveal all
symbols through position $i$. The time $T_j\eqdef S_j+t-1$ is a stopping
time. Indeed, if the full-window maximum is at least $t$, it is also
the longest admissible match, and its decision requires only its first
mismatch symbol. If that maximum is shorter than $t$, at most a fixed
$t$-symbol prefix is needed to resolve all shorter classes. A long
phrase ending immediately before $i$ is resolved by symbol $i$; a
shorter phrase beginning at $s<i$ is resolved by symbols through
$s+t-1\le i+t-1$. Induction makes $\{S_j=i\}$ measurable in
$\mathcal F_{i+t-1}$, proving the stopping-time assertion.

Condition on $\mathcal F_{T_j}$. For any fixed distance $d\le W$, a
match of length $r>t$ requires $r-t$ fresh iid symbols to equal
previously determined symbols. Overlap does not change the probability:
when each new symbol is exposed, its required predecessor has already
been determined. A union bound over distances yields
\[
 \Pr\{L_j\ge r\mid\mathcal F_{T_j}\}
 \le\min\{1,Wq^{-(r-t)}\},\qquad r>t.
\]
Summing these tails gives
\begin{equation}
 \mathbb E[L_j\mid\mathcal F_{T_j}]
 \le\mu_B\eqdef t+\lceil B/h\rceil+(q-1)^{-1}
 =B/h+O(1).
 \label{eq:obstruction-mean}
\end{equation}
Let $N$ count infinite-parse starts in $[W+1,n]$. Truncating their
lengths at the file boundary makes the sum exactly $n-W$.
Multiplying the conditional bound by
the indicator that the corresponding start lies in the file and summing
at most $n-W$ terms gives
\begin{equation}
 \mathbb E N\ge(n-W)/\mu_B.
 \label{eq:obstruction-count}
\end{equation}
This step does not assume independent phrase lengths.

Choose a constant $0<\theta<1$, and put $m=\lfloor\theta B/h\rfloor$.
A past window contains $\lfloor W/m\rfloor$ independent disjoint
length-$m$ blocks. The probability that it fails to contain every one
of the $q^m$ possible words is at most
\begin{equation}
 \delta_B\eqdef q^m\exp\{-\lfloor W/m\rfloor q^{-m}\}.
 \label{eq:coverage-failure}
\end{equation}
This decreases faster than every inverse power of $B$.
At a position where full coverage holds and at least $m$ symbols remain,
the greedy phrase has length at least $m$. For large $B$ this is in the
full-window class and costs at least
\[
 J_B\eqdef B+c\log(m+1)-C=B+c\log B-O(1).
\]
The finite parser agrees with these infinite-parse starts until its
terminal region: a match truncated to length at least $t$ remains legal,
and any discrepancy below $t$ occurs in the final $t-1$ symbols.
Consequently every covered infinite-parse start with at least $m\ge t$
symbols remaining yields a finite-file phrase with the stated cost.
The number of phrase starts without coverage is bounded by the number
of such source positions, whose expectation is at most
$(n-W)\delta_B$. At most $m$ additional starts lie in the terminal
region. Therefore
\[
 \mathbb E C_B\ge J_B\left[(n-W)/\mu_B-(n-W)\delta_B-m\right].
\]
Divide by $n$, let $n\to\infty$, and expand $J_B/\mu_B$ to obtain
\eqref{eq:obstruction}.

For the permitted literal fallback, choose $h/b<\theta<1$.
Then $bm>J_B$ for large $B$, so both possible representations of a
covered long phrase have cost at least $J_B$. The same proof applies.
\hfill$\square$

The theorem rules out redundancy $O(B^{-p})$ for $p>1$, and even
$O(B^{-1})$, in this model. Its assumption on length costs holds for
common universal integer codes but does not follow from prefix
decodability alone: a length code can give a short codeword to the
typical match length, as the cap token of
Theorem~\ref{thm:fixed-field-separation} does. A raw block coded at $h$
bits per symbol on this uniform source also escapes the bound, as it
should. Minimum-bit parsing, decisions with unbounded lookahead, and
reparsing after a rejected long match are not covered. A better exponent
for practical WLZ would thus need a proof about its model for long
phrases and its parsing policy, not merely about its short-match
windows.
\section{Proof of the Significant-Bit Coding Bound}
\label{app:significant-bit-locality}

We specialize the uniform-prefix mixture argument of
Ryabko, Astola, and Egiazarian~\cite[Thm.~1 and Appendix]{WLZRyabko2003}.
For a fixed partition, $\Delta=\mathbb E r_J-H(X\mid J)$ is convex
in the distribution, since conditional entropy is concave
\cite[Thm.~2.1]{WLZSaid2004}; it is nonnegative since a bucket has at
most $2^{r_J}$ valid values. Separate the exact head and cap atoms.
Set $a_m\eqdef p_m$ for $k\le m<c$ and $a_c\eqdef0$.
The tail is a mixture of uniforms $U[k,m]$, with
\[
 w_m\eqdef(m-k+1)(a_m-a_{m+1})\ge0,
 \qquad \sum_{m=k}^{c-1}w_m=\tau.
\]
Each exact atom has zero loss. For $U[k,m]$, the loss is zero if $m<T$.
Otherwise all occupied buckets except the last are complete. If that
last bucket has capacity $s$ and $t$ occupied values, then
$s=2^{\lfloor\log m\rfloor-\nu+1}\le2m/T$ and
\[
 \Delta(U[k,m])
 =\frac{t}{m-k+1}\log\frac{s}{t}
 \le\frac{s}{e\ln2\,(m-k+1)}
 \le\frac{2}{e\ln2\,(T-k+1)}.
\]
Here $\max_{0<z\le1}z\log(1/z)=1/(e\ln2)$ and
$m/(m-k+1)\le T/(T-k+1)$ for $m\ge T$.
Convexity gives~\eqref{eq:significant-bit-locality}, including a
partially used last bucket immediately before the exact cap.
Huffman coding of labels costs less than $H(J)+1$ bits per value;
whole-value Huffman coding costs at least $H(X)$ and is no worse than
the grouped prefix code. Use zero-bit singleton codes in both cases.
Summing over the three empirical streams gives
\eqref{eq:significant-bit-streams}.

For the stated Zstandard comparison, a uniform positive prefix $U[1,m]$
ending in a dyadic bucket $[s,2s-1]$ has loss
$t\log(s/t)/m\le1/(e\ln2)$, since $s\le m$.
Mixtures give the same bound for decreasing positive probabilities;
an arbitrary atom at zero contributes no loss. Refining a bucket
cannot increase the grouping loss, by nonnegativity of relative entropy
between its refined-label probabilities and relative bucket sizes.
The cited Zstandard tables are such refinements, including the ordinary
dyadic coded-offset table. This argument concerns the emitted numeric
histograms, not an assumed source distribution. \hfill$\square$

\section{Exact Minimum-Cost Parsing}
\label{app:exact-parsing}

The main results concern exact greedy parsing and the minimum token cost
$C_W^*$ as $n\to\infty$. This appendix concerns $C_W^*(x_1^n)$ for a given
input: which single-window shortcuts fail under length-dependent windows,
and how quickly $C_W^*$ and a minimizing parse can be computed. A parse is
a sequence of phrases, each a literal of cost $b$ or an admissible match
of length at most a cap $L$; phrase counts include literals. A phrase of
length $\ell$ starting at position $i$ ends before $i+\ell$.

\subsection{Two single-window properties that fail}

Under one window, a suffix of an admissible match is an admissible match
at the same distance. If every length from one is supported, exact greedy
parsing then minimizes the number of phrases. Let greedy and another parse
have their $t$th phrases at $g_t\ge p_t$. Either $g_t\ge p_{t+1}$, or the
other parse's $t$th phrase covers $g_t$ and its suffix from $g_t$ is a
phrase ending at $p_{t+1}$; in both cases $g_{t+1}\ge p_{t+1}$
(see~\cite{WLZCrochemore2014} for the general closure conditions).
Likewise, for match costs nondecreasing in length, some minimum-cost parse
uses only maximal matches, the longest of each cost from each
position~\cite{WZIPBitOptimal2013}. Both arguments shorten a match from the
left. Under length-dependent windows the shortened match may fall below
$t_B$ and lie outside its own window, the loss of suffix closure that
Theorem~\ref{thm:block-transfer} accounts for, and both properties
fail.

\emph{Example 1 (phrase count).} Let lengths at most 2 have window 4 and
longer lengths the whole history. After the history
\texttt{abcx\#xyz\#}, exact greedy parsing codes \texttt{abcxyz} as
\texttt{abcx}, \texttt{y}, \texttt{z}: the longest admissible match
\texttt{abcx} leaves \texttt{yz}, whose only earlier copy lies 7 symbols
back. The parse \texttt{abc}, \texttt{xyz} has two phrases.

\emph{Example 2 (maximal matches).} Let $x_1^{12}=\texttt{010000001000}$,
$L=8$, window 1 for lengths 2--4 and window 8 for lengths 5--8, literal
cost $b=2$, and match cost 7. Three literals, the match $(4,1)$ and the
match $(5,7)$ cost 20 bits, the minimum. Keeping only the longest
admissible match at each position, here $(5,1)$ at position 4, leaves
\texttt{1000}, whose only earlier copy lies at distance 7 and is
admissible only for lengths of at least 5; the best parse from literals
and longest matches costs 21 bits. With window 8 for every length, both
minima are 20 bits. Figure~\ref{fig:mwparse} shows these parses.

\begin{figure}[t]
\centering
\begin{tikzpicture}[font=\small,
  cell/.style={draw, minimum width=6mm, minimum height=6mm, inner sep=0pt, font=\ttfamily,
    text height=1.6ex, text depth=0.3ex},
  lit/.style={draw, fill=black!12},
  mat/.style={draw},
  row/.style={anchor=east, inner sep=0pt, align=right, font=\footnotesize},
  cost/.style={anchor=west, inner sep=0pt}]
\foreach \c [count=\i] in {0,1,0,0,0,0,0,0,1,0,0,0} {
  \node[cell] at (0.6*\i,0) {\c};
  \node[font=\scriptsize] at (0.6*\i,0.5) {\i};}
\newcommand{\phr}[5]{\draw[#5] (0.6*#2-0.3,#1-0.28) rectangle (0.6*#2+0.6*#3-0.3,#1+0.28);
  \node[font=\footnotesize, text height=1.5ex, text depth=0.3ex] at (0.6*#2+0.3*#3-0.3,#1) {#4};}
\node[row] at (0.0,-0.9) {minimum cost};
\phr{-0.9}{1}{1}{\texttt{0}}{lit} \phr{-0.9}{2}{1}{\texttt{1}}{lit} \phr{-0.9}{3}{1}{\texttt{0}}{lit}
\phr{-0.9}{4}{4}{$(4,1)$}{mat} \phr{-0.9}{8}{5}{$(5,7)$}{mat}
\node[cost] at (7.75,-0.9) {20 bits};
\node[row] at (0.0,-1.65) {longest match\\only};
\phr{-1.65}{1}{1}{\texttt{0}}{lit} \phr{-1.65}{2}{1}{\texttt{1}}{lit} \phr{-1.65}{3}{1}{\texttt{0}}{lit}
\phr{-1.65}{4}{5}{$(5,1)$}{mat}
\phr{-1.65}{9}{1}{\texttt{1}}{lit} \phr{-1.65}{10}{1}{\texttt{0}}{lit}
\phr{-1.65}{11}{1}{\texttt{0}}{lit} \phr{-1.65}{12}{1}{\texttt{0}}{lit}
\node[cost] at (7.75,-1.65) {21 bits};
\node[row] at (0.0,-2.55) {one window,\\longest match};
\phr{-2.55}{1}{1}{\texttt{0}}{lit} \phr{-2.55}{2}{1}{\texttt{1}}{lit} \phr{-2.55}{3}{1}{\texttt{0}}{lit}
\phr{-2.55}{4}{5}{$(5,1)$}{mat} \phr{-2.55}{9}{4}{$(4,7)$}{mat}
\node[cost] at (7.75,-2.55) {20 bits};
\end{tikzpicture}
\caption{\label{fig:mwparse}Example~2 as parses. With window~1 for lengths
2--4 and window~8 for lengths 5--8 (literals 2 bits, shaded; matches
7 bits), the admissible lengths at position~4 are $[2,4]\cup[5,5]$
(Lemma~\ref{lem:admissible-intervals}). The minimum-cost parse uses the
shorter match $(4,1)$; taking the longest match $(5,1)$ leaves
\texttt{1000}, whose only copy, at distance~7, is admissible only for
lengths of at least~5. With one window of 8 for every length, the
longest-match parse is again optimal.}
\end{figure}

\subsection{Length-only token costs}

Under the cost model~\eqref{eq:wzip-cost}, a match costs
$c(\ell)=a(\ell)+w(\ell)$, independent of its distance. The schedule need
not be monotone. Partition the supported lengths into the $K$ maximal
intervals $[s_k,e_k]$ on which $W(\ell)$ is constant, equal to $W_k$, with
$e_K=L$; for a nondecreasing schedule, $K$ is the number of distinct
windows. For a position $i$, let
$\lambda_k(i)$ be the largest $\ell\le\min\{L,n-i+1\}$ such that
$x_i^{i+\ell-1}$ has a copy at a distance
$d\le\min\{W_k,i-1\}$, or $0$ if there is none.

\begin{lemma}[Admissible lengths]
\label{lem:admissible-intervals}
The admissible match lengths at position $i$ are exactly
\[
 \bigcup_{k=1}^K\,[s_k,\min\{e_k,\lambda_k(i)\}].
\]
\end{lemma}

\noindent\emph{Proof.}
If $s_k\le\ell\le\min\{e_k,\lambda_k(i)\}$, the first $\ell$ symbols of
the copy defining $\lambda_k(i)$ form a length-$\ell$ match at a distance
at most $W_k=W(\ell)$. Conversely, an admissible length-$\ell$ match with
$\ell\in[s_k,e_k]$ has distance at most $W_k$, so
$\lambda_k(i)\ge\ell$. \hfill$\square$

\begin{theorem}[Exact parsing with length-only costs]
\label{thm:exact-parsing}
Let $Y$ be the number of maximal runs on which $a(\ell)$ is constant,
over the supported lengths up to $\min\{L,n\}$. For a finite alphabet,
$C_W^*(x_1^n)$ and a minimizing parse can be computed offline in
$O(n(K+Y)\log n)$ time using $O(n)$ words of space.
\end{theorem}

\noindent\emph{Proof.}
Build the suffix array of $x_1^n$, its inverse, and the longest-common-prefix
array with a constant-time range-minimum structure, all in $O(n)$
time~\cite{WLZSuffixArray2006,WLZLCP2001,WLZRMQ2000}; the longest common
prefix of any two suffixes is then available in $O(1)$ time. Comparing
suffixes also covers overlapping copies. Process $i=1,\ldots,n$ in order.
A maximum segment tree over suffix ranks stores each earlier position at
its suffix's rank. For each $k$, descending this tree finds the nearest
ranks below and above that of suffix $i$ whose stored positions are at
least $i-W_k$. Among the positions within distance $W_k$, the suffix
sharing the longest prefix with suffix $i$ is at one of these two ranks,
so $\lambda_k(i)$, capped at $\min\{L,n-i+1\}$, takes $O(\log n)$ time.
Position $i$ is then inserted.

Let $f(i)$ be the minimum cost of parsing $x_1^{i-1}$, with $f(1)=0$; the
target is $f(n+1)$. All phrases ending at $i-1$ start earlier, so $f(i)$
is final when position $i$ is processed. Relax $f(i+1)$ by $f(i)+b$. The
class boundaries $s_k$ and the run boundaries of $a$ cut the supported
lengths into at most $K+Y$ segments of constant cost. By
Lemma~\ref{lem:admissible-intervals}, the admissible lengths at $i$ meet
each segment in an interval $[u,v]$, possibly empty; relax $f(j)$ by
$f(i)+c(u)$ for all $j\in[i+u,i+v]$. A segment tree with lazy range
minimum updates and point queries performs each relaxation and each read
of $f(i)$ in $O(\log n)$ time, and stores with each value the phrase that
attains it, for backtracking. \hfill$\square$

The algorithm works offline on the whole input with suffix-array
structures, separately from the hash indexes of
Section~\ref{sec:hash-search}. For $K=1$ it is the single-window
computation; length-dependent windows add at most $K-1$ segments per
position. Under the hypotheses of Theorem~\ref{thm:greedy-universality},
all lengths of at least $t_B$ share the full window, so $K\le t_B$. The
code-cost bound there limits $a(\ell)$ but not the number $Y$ of its
changes. For a length code that is constant on dyadic length ranges,
$Y=O(\log L)$; with $L=W$ the time is $O(n(t_B+B)\log n)=O(nB\log n)$
for $t_B=o(B)$, and the minimum-bit variant of the entropy
estimator~\eqref{eq:wlz-entropy-estimator}, which evaluates $C^*_{B,n}$
for every $B\le K_n\le\log n$, takes $O(n\log^3n)$ time.

\subsection{Distance-dependent costs}

When a match cost $c(\ell,d)$ also depends on the distance, as with
variable-length distance codes, a constant-cost segment of lengths no
longer determines the cost. Example~2 shows that maximal matches do not
suffice. For fixed costs that are nondecreasing in match length, a
bounded number of shorter matches near each maximal one does. Arbitrary
fitted Huffman costs need not be monotone, and neither is the cheap cap
token of Theorem~\ref{thm:fixed-field-separation}; such costs need
separate treatment, although with length-only costs they remain covered
by Theorem~\ref{thm:exact-parsing}.

\begin{proposition}[Pruned parsing graph]
\label{prop:pruned-graph}
Suppose that every length of at least $t_B$ has the full window, that
$c(\ell,d)\ge0$ is nondecreasing in $\ell$ for each $d$, and that $b\ge0$.
At each position $i$, keep
\begin{enumerate}
\item[(a)] for each $\ell<t_B$, one cheapest admissible match of
length $\ell$;
\item[(b)] for each cost $q$ of a match starting at $i$, a longest
admissible match $(\ell',d')$ of cost at most $q$, and its prefixes at
distance $d'$ of every length $\ell\ge t_B$ with
$\ell'-t_B<\ell\le\ell'$.
\end{enumerate}
Some minimum-cost parse uses only literals and kept matches. With at most
$Q$ distinct costs per position, at most $t_B-1+Qt_B$ matches are kept per
position.
\end{proposition}

\noindent\emph{Proof.}
Prefixes of length at least $t_B$ are admissible, and by monotonicity a
prefix in (b) costs at most $q$. Take a minimum-cost parse whose matches
before $i$ are kept, and let $e$ be its match at $i$, of length $\ell$ and
cost $q$. If $\ell<t_B$, replace $e$ by the kept match of length $\ell$.
Otherwise let $(\ell',d')$ be the match of (b) for $q$, so
$\ell'\ge\ell\ge t_B$; if $\ell'=\ell$, replace $e$ by it. If
$\ell'>\ell$, let $f$ be the phrase covering position $i+\ell'-1$; it
starts at some $k\ge i+\ell$ and ends before $m$. If $m=i+\ell'$ or
$m-i-\ell'\ge t_B$, replace the phrases from $e$ through $f$ by
$(\ell',d')$, followed, if $m>i+\ell'$, by the suffix of $f$ from
$i+\ell'$: a match at the distance of $f$, of length at least $t_B$ and
cost at most that of $f$. If $0<m-i-\ell'<t_B$ and $f$ has length at least
$t_B$, put $j=m-t_B$, so that $k\le j<i+\ell'$; replace the phrases from
$e$ through $f$ by the prefix of $(\ell',d')$ ending before $j$ and the
length-$t_B$ suffix of $f$ from $j$. If $0<m-i-\ell'<t_B$ and $f$ is
shorter than $t_B$, then $k>m-t_B>i+\ell'-t_B$; replace the phrases from
$e$ up to $f$ by the prefix of $(\ell',d')$ ending before $k$, and keep
$f$. Each prefix used has length between $\ell$ and $\ell'$ and more than
$\ell'-t_B$, so it is kept. The removed phrases have nonnegative cost, so
the parse remains minimal, and its match at $i$ is now kept. Induction on
$i$ completes the proof. \hfill$\square$

For $t_B=1$, (a) is empty and (b) keeps exactly the maximal matches of
the single-window case~\cite{WZIPBitOptimal2013}. Under length-dependent
windows the graph grows by a factor of at most $t_B$, which is $o(B)$
under the hypotheses of Theorem~\ref{thm:greedy-universality}. Generating
the kept matches on the fly, as~\cite{WZIPBitOptimal2013} does for one
window, is not addressed here; the per-cost-class sliding-window suffix
trees of~\cite{WLZCrochemoreREP2013} are a natural starting point. Neither result covers repeat-offset caches
or adaptively updated prices, whose costs depend on the parse history.

\end{document}